\documentclass[preprint,11pt]{elsarticle}

\usepackage{amssymb}
\usepackage{a4wide}

\journal{}

\usepackage{microtype}
\usepackage{graphicx}
\usepackage[algoruled,linesnumbered,algo2e,vlined]{algorithm2e}
\usepackage[final,mode=multiuser]{fixme}
\usepackage{amsmath,amssymb}
\usepackage{url}
\usepackage{rotating}
\usepackage{comment}
\usepackage{multirow}
\usepackage{wrapfig}
\usepackage{amsthm}

\newtheorem{theorem}{Theorem}
\newtheorem{lemma}{Lemma}
\newtheorem{corollary}{Corollary}
\newtheorem{observation}{Observation}
\newtheorem{definition}{Definition}
\theoremstyle{definition}
\newtheorem{example}{Example}
\theoremstyle{remark}
\newtheorem*{remark*}{Remark}

\SetKw{To}{to}
\SetKw{Each}{each}

\newcommand{\DAWG}{\mathsf{DAWG}}
\newcommand{\CDAWG}{\mathsf{CDAWG}}

\newcommand{\DAWGedge}{\delta_{\mathsf{D}}}
\newcommand{\DAWGlink}{\mathit{sl_{\mathsf{D}}}}
\newcommand{\Longest}{\mathsf{long}}
\newcommand{\Shortest}{\mathsf{short}}

\newcommand{\ST}{\mathsf{STree}}
\newcommand{\STrie}{\mathsf{STrie}}

\newcommand{\SLT}{\mathsf{SLT}}

\newcommand{\Substr}{\mathsf{Substr}}
\newcommand{\Suffix}{\mathsf{Suffix}}
\newcommand{\MAW}{\mathsf{MAW}}

\newcommand{\AST}{\mathsf{STree^\prime}}
\newcommand{\AVs}{\mathsf{V_{S}^\prime}}
\newcommand{\AEs}{\mathsf{E_{S}^\prime}}

\newcommand{\LSTrie}{\mathsf{LSTrie}}

\newcommand{\BT}{\mathsf{BT}}

\newcommand{\ATree}{\mathsf{ATree}}
\newcommand{\Va}{\mathsf{V_A}}

\newcommand{\Ea}{\mathsf{E}_\mathsf{A}}
\newcommand{\Eaf}{\mathsf{E}_\mathsf{A}^{F}}
\newcommand{\Eab}{\mathsf{E}_\mathsf{A}^{B}}

\newcommand{\rev}[1]{\hat{#1}}

\newcommand{\Lt}{\mathsf{L_{T}}}
\newcommand{\Vt}{\mathsf{V_{T}}}
\newcommand{\Et}{\mathsf{E_{T}}}

\newcommand{\Ls}{\mathsf{L_{S}}}
\newcommand{\Vs}{\mathsf{V_{S}}}
\newcommand{\Es}{\mathsf{E_{S}}}

\newcommand{\Ld}{\mathsf{L_{D}}}
\newcommand{\Vd}{\mathsf{V_{D}}}
\newcommand{\Ed}{\mathsf{E_{D}}}

\newcommand{\WL}{\mathsf{WL}}
\newcommand{\mWL}{\mathsf{mWL}}

\newcommand{\BegPos}{\mathsf{BegPos}}
\newcommand{\EndPos}{\mathsf{EndPos}}

\newcommand{\nil}{\mathit{nil}}

\newcommand{\Leqr}{\equiv_\mathit{L}}
\newcommand{\Reqr}{\equiv_\mathit{R}}

\newcommand{\Leqc}[1]{[{#1}]_{L}}
\newcommand{\Reqc}[1]{[{#1}]_{R}}

\newcommand{\Lrep}[1]{\overrightarrow{#1}}
\newcommand{\Rrep}[1]{\overleftarrow{#1}}
\newcommand{\parent}{\mathsf{p}}

\fxuseenvlayout{colorsig}
\fxusetargetlayout{color}
\FXRegisterAuthor{yf}{ayf}{YF}
\FXRegisterAuthor{si}{asi}{SI}
\FXRegisterAuthor{hb}{ahb}{HB}
\FXRegisterAuthor{mt}{amt}{MT}
\fxsetface{env}{}

\begin{document}

\begin{frontmatter}

  \title{Linear-time Computation of DAWGs, Symmetric Indexing Structures, and MAWs for Integer Alphabets}

\author[kyushu,fujitsu]{Yuta Fujishige}
\author[kyushu-ko]{Yuki Tsujimaru}
\author[kyushu]{Shunsuke Inenaga}
\ead{inenaga.shunsuke.380@m.kyushu-u.ac.jp}
\author[tmdu]{Hideo Bannai}
\ead{hdbn.dsc@tmd.ac.jp}
\author[kyushu]{Masayuki Takeda}
\address[kyushu]{Department of Informatics, Kyushu University, Japan}
\address[fujitsu]{Fujistu Limited (current affiliation)}
\address[kyushu-ko]{Department of Electrical Engineering and Computer Science, Kyushu University, Japan}
\address[tmdu]{M\&D Data Science Center, Tokyo Medical and Dental University, Japan}

\begin{abstract}
  The \emph{directed acyclic word graph} (\emph{DAWG}) of a string $y$
  of length $n$ is the smallest (partial) DFA which recognizes all suffixes of $y$
  with only $O(n)$ nodes and edges.
  In this paper, we show how to construct the DAWG for the input string $y$
  from the suffix tree for $y$,
  in $O(n)$ time for \emph{integer alphabets of polynomial size in $n$}.
  In so doing, we first describe a folklore algorithm which, given
  the suffix tree for $y$, constructs
  the DAWG for the reversed string $\rev{y}$ in $O(n)$ time.
  Then, we present our algorithm that builds
  the DAWG for $y$ in $O(n)$ time for integer alphabets, from
  the suffix tree for $y$.
  We also show that a straightforward modification to our DAWG construction
  algorithm leads to the first $O(n)$-time algorithm
  for constructing the \emph{affix tree} of a given string $y$
  over an integer alphabet.
  Affix trees are a text indexing structure
  supporting bidirectional pattern searches.
  We then discuss how our constructions can lead to linear-time algorithms
  for building other text indexing structures,
  such as \emph{linear-size suffix tries} and \emph{symmetric CDAWGs}
  in linear time in the case of integer alphabets.
  As a further application to our $O(n)$-time DAWG construction algorithm,
  we show that the set $\MAW(y)$ of all \emph{minimal absent words} (\emph{MAWs})
  of $y$ can be computed in
  \emph{optimal}, input- and output-sensitive $O(n + |\MAW(y)|)$ time and
  $O(n)$ working space for integer alphabets.
\end{abstract}


\end{frontmatter}

\section{Introduction}

\subsection{Constructing DAWGs for integer alphabets}

Text indexes are fundamental data structures that allow for efficient processing of string data,
and have been extensively studied. Although there are several alternative data structures which can be
used as an index, such as suffix trees~\cite{Weiner73} and suffix arrays~\cite{ManberM93} (and their compressed versions~\cite{Sadakane03,GrossiV05,Sadakane07}),
in this paper, we focus on \emph{directed acyclic word graphs} (\emph{DAWGs})~\cite{Blumer85}, which are a fundamental data structure for string processing.
Intuitively, the DAWG of string $y$,
denoted $\DAWG(y)$, is an edge-labeled DAG obtained by
merging isomorphic subtrees of the trie representing
all suffixes of string $y$, called the suffix trie of $y$.
Hence, $\DAWG(y)$ can be seen as an automaton
recognizing all suffixes of $y$.
Let $n$ be the length of the input string $y$.
Despite the fact that the number of nodes and edges of the suffix trie
is $\Omega(n^2)$ in the worst case,
Blumer et al.~\cite{Blumer85} proved that, surprisingly,
$\DAWG(y)$ has at most $2n-1$ nodes and $3n - 4$ edges for $n > 2$.
Crochemore~\cite{Crochemore86} showed that
$\DAWG(y)$ is the smallest (partial)  automaton
recognizing all suffixes of $y$,
namely, the sub-tree merging operation
which transforms the suffix trie to $\DAWG(y)$ indeed minimizes the automaton.

Since $\DAWG(y)$ is a DAG,
more than one string can be represented by the same node in general.
It is known that every string represented by the same node of $\DAWG(y)$
has the same set of ending positions in the string $y$.
Due to this property, if $z$ is the longest string represented by
a node $v$ of $\DAWG(y)$, then any other string represented by
the node $v$ is a proper suffix of $z$.
Hence, the \emph{suffix link} of each non-root node of $\DAWG(y)$
is well-defined; if $ax$ is the shortest string represented
by node $v$ where $a$ is a single character and $x$ is a string,
then the suffix link of $ax$ points to the node of $\DAWG(y)$
that represents string $x$.

One of the most intriguing properties of DAWGs
is that the suffix links of $\DAWG(y)$ for any string $y$
forms the suffix tree~\cite{Weiner73} of the reversed string of $y$.
Hence, $\DAWG(y)$ augmented with suffix links can be seen
as a \emph{bidirectional} text indexing data structure.
This line of research was followed by other types of
bidirectional text indexing data structures such as
\emph{symmetric compact DAWGs} (\emph{SCDAWGs})~\cite{Blumer87}
and \emph{affix trees}~\cite{Stoye00,Maass03}.

\subsubsection{Our Contributions to Text Indexing Constructions}

Time complexities for constructing text indexing data structures
depend on the underlying alphabet.
See Table~\ref{table:results}.
For a given string $y$ of length $n$ over an ordered alphabet of size $\sigma$,
the suffix tree~\cite{McCreight76}, the suffix array~\cite{ManberM93},
the DAWG, and the \emph{compact DAWGs} (\emph{CDAWGs})~\cite{Blumer87}
of $y$ can all be constructed in $O(n \log \sigma)$ time.
Here, we recall that the CDAWG of a string $y$ is a path-compressed
version of the DAWG for $y$,
or equivalently, the CDAWG of $y$ is a DAG obtained by merging
isomorphic subtrees of the suffix tree of $y$.
The aforementioned bounds immediately lead to $O(n)$-time construction algorithms for constant-size alphabets.

\begin{table}[bt]
  \caption{Space requirements and construction times for text indexing structures
   for input strings of length $n$ over an alphabet of size $\sigma$. }
  \label{table:results}
  \centerline{
  \footnotesize
  \begin{tabular}{|l||l|l|l|l|}\hline
    & \multirow{2}{*}{\begin{tabular}[c]{@{}c@{}}space\\(in words) \end{tabular}} & \multicolumn{3}{|c|}{construction time}\\ \cline{3-5}
    &  & ordered alphabet &  integer alphabet & constant-size alphabet \\ \hline \hline
    suffix tries & $O(n^2)$ & $O(n^2)$ & $O(n^2)$ & $O(n^2)$ \\ \hline
    suffix trees & $O(n)$ & $O(n \log \sigma)$~\cite{McCreight76} & $O(n)$~\cite{Farach-ColtonFM00} & $O(n)$~\cite{Weiner73} \\ \hline
    suffix arrays & $O(n)$ & $O(n \log \sigma)$~\cite{McCreight76}+\cite{ManberM93} & $O(n)$~\cite{Farach-ColtonFM00}+\cite{ManberM93} & $O(n)$~\cite{Weiner73}+\cite{ManberM93} \\ \hline
    DAWGs & $O(n)$ & $O(n \log \sigma)$~\cite{Blumer85}& $O(n)$~[this work] & $O(n)$~\cite{Blumer85} \\ \hline
    CDAWGs & $O(n)$ & $O(n \log \sigma)$~\cite{Blumer87} & $O(n)$~\cite{Narisawa07}& $O(n)$~\cite{Blumer87} \\ \hline
    symmetric CDAWGs & $O(n)$ & $O(n \log \sigma)$~\cite{Blumer87} & $O(n)$~[this work] & $O(n)$~\cite{Blumer87} \\ \hline
    affix trees & $O(n)$ & $O(n \log \sigma)$~\cite{Maass03} & $O(n)$~[this work] & $O(n)$~\cite{Maass03} \\ \hline
    linear-size suffix tries & $O(n)$ & $O(n \log \sigma)$~\cite{CrochemoreEGM16} & $O(n)$~[this work] & $O(n)$~\cite{CrochemoreEGM16} \\ \hline
  \end{tabular}
  }
\end{table}

In this paper, we are particularly interested in
input strings of length $n$ over an \emph{integer alphabet} of polynomial size in $n$.
Farach-Colton et al.~\cite{Farach-ColtonFM00} proposed
the first $O(n)$-time suffix tree construction algorithm for
integer alphabets.
Since the out-edges of every node of the suffix tree
constructed by McCreight's~\cite{McCreight76} and
Farach-Colton et al.'s algorithms
are lexicographically sorted,
and since sorting is an obvious
lower-bound for constructing edge-sorted suffix trees,
the above-mentioned suffix-tree construction algorithms are optimal
for ordered and integer alphabets, respectively.
Since the suffix array of $y$ can be easily obtained
in $O(n)$ time from the edge-sorted suffix tree of $y$,
suffix arrays can also be constructed in optimal time.
In addition, since the edge-sorted suffix tree of $y$
can easily be constructed in $O(n)$ time from the edge-sorted CDAWG of $y$,
and since the edge-sorted CDAWG of $y$ can be constructed in $O(n)$ time
from the edge-sorted DAWG of $y$~\cite{Blumer87},
sorting is also a lower-bound for constructing
edge-sorted DAWGs and edge-sorted CDAWGs.
Using the technique of Narisawa et al.~\cite{Narisawa07},
edge-sorted CDAWGs can be constructed in optimal $O(n)$ time for
integer alphabets.
On the other hand,
the only known algorithm to construct DAWGs was Blumer et al.'s
$O(n \log \sigma)$-time online algorithm~\cite{Blumer85}
for ordered alphabets of size $\sigma$,
which results in $O(n \log n)$-time DAWG construction for integer alphabets.
In this paper, we show that the gap between the upper and lower bounds
for DAWG construction can be closed,
by introducing how to construct edge-sorted DAWGs in $O(n)$ time for integer alphabets,
in two alternative ways.

It is known that the suffix tree can be augmented with \emph{Weiner links},
which are a generalization of the reversed suffix links.
The DAG consisting of the nodes of the suffix tree for a string $y$ and its Weiner links
coincides with the DAWG of the reversed string $\rev{y}$~\cite{Blumer85,Chen85}.
In this paper, we first describe an $O(n)$-time folklore algorithm which computes the
sorted Weiner links
for integer alphabets, given that the suffix tree of the string is already computed.
This immediately gives us an $O(n)$-time algorithm for constructing
the edge-sorted DAWG for the reversed string $\rev{y}$ over an integer alphabet.

It was still left open whether one could efficiently construct the DAWG for a string $y$
from the suffix tree for $y$, in the case of integer alphabets.
We close this question by proposing an $O(n)$-time algorithm
that builds the edge-sorted DAWG for the input (forward) string $y$.
Our algorithm builds $\DAWG(y)$ for a given string $y$
by transforming the suffix tree of $y$ to $\DAWG(y)$.
In other words, our algorithm simulates the minimization
of the suffix trie of $y$ to $\DAWG(y)$ using only $O(n)$ time and space.
Our algorithm also computes the suffix links of the DAWG,
and can thus be applied to various kinds of string processing problems.
This also means that we can construct the suffix tree for the reversed string $\rev{y}$,
as a byproduct.

A simple modification to our $O(n)$-time DAWG construction algorithm
also leads us to the first $O(n)$-time algorithm to construct edge-sorted affix trees
for integer alphabets.
We remark that the previous best known affix-tree construction
algorithm of Maa\ss~\cite{Maass03} requires
$O(n \log n)$ time for integer alphabets.

In addition, we show that our construction algorithms
for DAWGs and affix trees lead to linear-time constructions
of other indexing structures such as
symmetric CDAWGs~\cite{Blumer87} and linear-size suffix tries~\cite{CrochemoreEGM16} in the case of integer alphabets.

\subsection{Computing Minimal Absent Words for Integer Alphabets}

As yet another application of our $O(n)$-time DAWG construction algorithm,
we present an optimal time algorithm
to compute \emph{minimal absent words} for a given string.
  There are a number of applications to minimal absent words,
  including data compression~\cite{DBLP:conf/icalp/CrochemoreMRS99, DBLP:conf/isita/OtaM14a}
and molecular biology~\cite{DBLP:conf/psb/HampikianA07,DBLP:journals/bmcbi/HeroldKG08, DBLP:journals/ipl/WuJS10, DBLP:journals/bioinformatics/SilvaPCPF15,Charalampopoulos18, DBLP:journals/almob/AlmirantisCGIMP17,DBLP:journals/bioinformatics/HeliouPP17}.
Hence, it is important to develop
efficient algorithms for computing minimal absent words from a given string.

Let $\MAW(y)$ be the set of minimal absent words of a string $y$.
Crochemore et al.~\cite{crochemore98:_autom_forbid_words} proposed
an algorithm to compute $\MAW(y)$
in $\Theta(n \sigma)$ time and $O(n)$ working space.
Their algorithm first constructs $\DAWG(y)$ with suffix links
in $O(n \log \sigma)$ time and computes $\MAW(y)$ in $\Theta(n\sigma)$ time
using $\DAWG(y)$ and its suffix links.
Since $|\MAW(y)| = O(n \sigma)$~\cite{MignosiRS02},
the output size $|\MAW(y)|$ is hidden in the running time of their algorithm.

Later, Barton et al. proposed an alternative algorithm that computes $\MAW(y)$
using the suffix array, in $\Theta(n\sigma)$ time and space~\cite{DBLP:journals/bmcbi/BartonHMP14}.
The algorithm presented in \cite{DBLP:conf/ppam/BartonHMP15} can be seen as
a parallel version of this Barton et al.'s algorithm using suffix arrays.
While all of these algorithms achieve $O(n)$ running time
in the case of constant-size alphabets ($\sigma \in O(1)$),
it can require $O(n^2)$ time when $\sigma \in \Omega(n)$.

Mignosi et al.~\cite{MignosiRS02} gave a tight bound on $|\MAW(y)|$
such that $\sigma \leq |\MAW(y)| \leq (\sigma_y - 1)(|y|-1) + \sigma$,
where $\sigma$ is the size of the underlying alphabet
and $\sigma_y$ is the number of distinct characters occurring in $y$.
However, $|\MAW(y)|$ can be $o(n\sigma)$ for many strings.
Thus, it is important to develop an algorithm whose running time is
\emph{output-sensitive}, namely, linear in the output size $|\MAW(y)|$.

In this paper, we show that $\MAW(y)$ can be computed in
\emph{optimal} $O(n + |\MAW(y)|)$ time given that
the edge-sorted DAWG of $y$ is already computed.
In more detail,
given the edge-sorted DAWG for $y$,
which can be computed in $O(n)$ time as above,
our algorithm computes $\MAW(y)$ in $O(n + |\MAW(y)|)$ optimal time.
We remark that our algorithm for computing $\MAW(y)$ itself works within this time for general ordered alphabets.

Our algorithm is a modification of
Crochemore et al.'s algorithm~\cite{crochemore98:_autom_forbid_words}
for finding $\MAW(y)$.
We emphasize that for non-constant-size alphabets Crochemore et al.'s algorithm
takes super-linear time in terms of the input string length
independently of the output size $|\MAW(y)|$,
and thus our result greatly improves the efficiency for integer alphabets.

\subsubsection{Other efficient algorithms for computing MAWs}
Belazzougui et al.~\cite{BelazzouguiCKM13} showed that
using a representation of the bidirectional BWT of the input string $y$ of length $n$,
$\MAW(y)$ can be computed in $O(n + |\MAW(y)|)$ time.
Although the construction time for the representation
of the bidirectional BWT was not given in~\cite{BelazzouguiCKM13},
it was later shown by Belazzougui and Cunial~\cite{BelazzouguiC19}
that the bidirectional BWT of a given string $y$ over an integer alphabet
can be incrementally constructed in $O(n)$ time.

Independently of our work,
Charalampopoulos et al.~\cite{DBLP:conf/spire/Charalampopoulos18}
proposed an $O(n)$-space data structure
which, given a positive integer $l$, computes the set $\MAW_l(y)$ of
all minimal absent words of length $l$
in $O(1+ |\MAW_l(y)|)$ query time.
This data structure can be constructed in $O(n)$ time and space.
By querying this data structure for all possible $l = 1, \ldots, n+1$,
one can compute $\MAW(y) = \bigcup_{l=1}^{n+1} \MAW_{l}(y)$ 
in $O(n + |\MAW(y)|)$ total time.

Fici and Gawrychowski~\cite{FiciG19} showed an optimal
$O(N+|\MAW(T)|)$-time algorithm for computing all MAWs
for a rooted tree $T$ of size $N$ in the case of integer alphabets.
It is known that the DAWG of a tree of size $N$
can have $\Omega(N^2)$ edges while its number of nodes is still $O(N)$~\cite{Inenaga20,Inenaga21}.
Instead of explicitly building the DAWG for $T$,
the algorithm of Fici and Gawrychowski~\cite{FiciG19}
simulates DAWG transitions by cleverly using lowest common ancestor queries
on the suffix tree of the reversed input tree $\rev{T}$.
In this view, their method can be seen as a tree input version
of our linear-time algorithm for computing MAWs.

Recently, Akagi et al.~\cite{AkagiOMNI22} proposed a space-efficient
representation of all MAWs for the input string $y$ which is
based on the \emph{run length encoding} (\emph{RLE}).
Their data structure takes $O(m)$ space and can output all MAWs in $O(|\MAW(y)|)$ time, where $m$ is the RLE size of $y$.

\subsection*{New Materials}
A preliminary version of this paper appeared in~\cite{DBLP:conf/mfcs/FujishigeTIBT16}.
Below is a list of the new materials in this full version:
\begin{itemize}
\item A clear description of the folklore linear-time construction of $\DAWG(\rev{y})$ for the reversed input string $\rev{y}$ from the suffix tree for the input string $y$ (Section~\ref{sec:folklore_algo});

\item The observation that our linear-time affix tree construction algorithm can also build $\DAWG(\rev{y})$ for the reversed input string $\rev{y}$ (Section~\ref{sec:DAWG_reversed});

  \item The first linear-time construction of symmetric CDAWGs in the case of integer alphabets (Section~\ref{sec:symmetric_cdawg});

\item The first linear-time construction of linear-size suffix tries in the case of integer alphabets (Section~\ref{sec:linear-size_suffix_tries}).
\end{itemize}

\section{Preliminaries}\label{preliminaries}

\subsection{Strings}
Let $\Sigma$ denote the alphabet.
An element of $\Sigma ^*$ is called a {\em string}.
Let $\varepsilon$ denote the empty string,
and let $\Sigma^+ = \Sigma^* \setminus \{\varepsilon\}$.
For any string $y$, we denote its length by $|y|$.
For any $1\leq i \leq |y|$, we use $y[i]$ to denote the
$i$th character of $y$.
For any string $y$, let $\rev{y}$ denote the reversed string of $y$.
If $y = uvw$ with $u,v,w \in \Sigma^*$,
then $u$, $v$, and $w$ are said to be
a \emph{prefix}, \emph{substring}, and \emph{suffix} of $y$,
respectively.
For any $1 \leq i \leq j \leq |y|$,
$y[i..j]$ denotes the substring of $y$
which begins at position $i$ and ends at position $j$.
For convenience, let $y[i..j] = \varepsilon$ if $i > j$.
Let $\Substr(y)$ and $\Suffix(y)$ denote
the set of all substrings and that of all suffixes of $y$, respectively.

Throughout this paper, we will use $y$ to denote the input string.
For any string $x \in \Sigma^*$, we define
the sets of beginning and ending positions of occurrences of $x$ in $y$,
respectively, by
\begin{eqnarray*}
  \BegPos(x) & = & \{ i \mid i\in [1,|y|-|x|+1],y[i..i+|x|-1] = x \}, \\
  \EndPos(x) & = & \{ i \mid i\in [|x|,|y|],y[i-|x|+1..i] = x \}.
\end{eqnarray*}
For any strings $u,v$, we write
$u \Leqr v$ (resp. $u\Reqr v$) when $\BegPos(u)=\BegPos(v)$ (resp. $\EndPos(u)=\EndPos(v)$).
For any string $x\in \Sigma^*$, the equivalence classes with respect to
$\Leqr$ and $\Reqr$ to which $x$ belong, are respectively denoted by
$\Leqc{x}$ and $\Reqc{x}$.
Also, $\Lrep{x}$ and $\Rrep{x}$ respectively denote the longest elements
of $\Leqc{x}$ and $\Reqc{x}$.


For any set $S$ of strings where no two strings are of the same length,
let $\Longest(S) = \arg \max\{|x| \mid x \in S\}$ and
$\Shortest(S) = \arg \min\{|x| \mid x \in S\}$.

In this paper, we assume that the input string $y$ of length $n$
is over the integer alphabet $[1, n^c]$ for some constant $c$, and that
the last character of $y$ is a unique character denoted by $\mathtt{\$}$
that does not occur
elsewhere in $y$.
Our model of computation is a standard word RAM of machine word size
$\log_2 n$.
Space complexities will be evaluated by the number of words (not bits).

\subsection{Suffix trees and DAWGs}

Suffix trees~\cite{Weiner73} and
directed acyclic word graphs (\emph{DAWGs})~\cite{Blumer85}
are fundamental text data structures.
Both of these data structures are based on suffix tries.
The \emph{suffix trie} for string $y$,
denoted $\STrie(y)$, is a trie representing $\Substr(y)$,
formally defined as follows.
\begin{definition}
  $\STrie(y)$ for string $y$ is an edge-labeled
  rooted tree $(\Vt,\Et)$ such that
  \begin{eqnarray*}
   \Vt & = & \{x \mid x \in \Substr(y)\} \\
   \Et & = & \{(x, b ,xb) \mid x,xb  \in \Vt, b \in \Sigma\}.
  \end{eqnarray*}
  The second element $b$ of each edge $(x, b ,xb)$ is the label of the edge.
  We also define the set $\Lt$ of labeled ``reversed'' edges called the {\em suffix links} of $\STrie(y)$
  by
  \[
   \Lt = \{(ax, a ,x) \mid x, ax \in \Substr(y),  a \in \Sigma\}.
  \]
\end{definition}
As can be seen in the above definition,
each node $v$ of $\STrie(y)$ can be identified with the substring of $y$
that is represented by $v$.
Assuming that string $y$ terminates with a unique character
that appears nowhere else in $y$,
for each non-empty suffix $y[i..|y|] \in \Suffix(y)$
there is a unique leaf $\ell_i$ in $\STrie(y)$
such that the suffix $y[i..|y|]$ is spelled out by the path
from the root to $\ell_i$.

It is well known that $\STrie(y)$ may require $\Omega(n^2)$ space.
One idea to reduce its space to $O(n)$ is to contract
each path consisting only of non-branching edges
into a single edge labeled with a non-empty string.
This leads to the suffix tree $\ST(y)$ of string $y$.
Following conventions from~\cite{Blumer87,InenagaHSTAMP05},
$\ST(y)$ is defined as follows.

\begin{definition}
  $\ST(y)$ for string $y$ is an edge-labeled
  rooted tree $(\Vs,\Es)$ such that
  \begin{eqnarray*}
   \Vs & = & \{\Lrep{x} \mid x\in \Substr(y)\} \\
   \Es & = & \{(x, \beta ,x\beta) \mid x,x\beta \in \Vs, \beta \in \Sigma^+, b=\beta[1],   \Lrep{xb} = x\beta \}.
  \end{eqnarray*}
  The second element $\beta$ of each edge $(x, \beta, x\beta)$ is the label of the edge.
  We also define the set $\Ls$ of labeled ``reversed'' edges called the {\em suffix links} of $\ST(y)$
  by
  \[
   \Ls = \{(ax, a, x) \mid x, ax \in \Vs,  a \in \Sigma\},
  \]
  and denote the tree $(\Vs,\Ls)$ of the suffix links by $\SLT(y)$.
\end{definition}
Observe that each internal node of $\ST(y)$ is a branching internal node in $\STrie(y)$.
Note that for any $x \in \Substr(y)$
the leaves in the subtree rooted at $\Lrep{x}$ correspond to $\BegPos(x)$.
By representing each edge label $\beta$ with a pair of
integers $(i, j)$ such that $y[i..j] = \beta$,
$\ST(y)$ can be represented with $O(n)$ space.


An alternative way to reduce the size of $\STrie(y)$ to $O(n)$
is to regard $\STrie(y)$ as a partial DFA
which recognizes $\Suffix(y)$, and to minimize it.
This leads to the directed acyclic word graph $\DAWG(y)$ of string $y$.
Following conventions from~\cite{Blumer87,InenagaHSTAMP05},
$\DAWG(y)$ is defined as follows.

\begin{definition}
  $\DAWG(y)$ of string
  $y$ is an edge-labeled DAG $(\Vd,\Ed)$ such that
  \begin{eqnarray*}
   \Vd & = & \{\Reqc{x} \mid x \in \Substr(y)\} \\
   \Ed & = & \{(\Reqc{x}, b, \Reqc{xb}) \mid x, xb \in \Substr(y), b \in \Sigma\}.
  \end{eqnarray*}
   We also define the set $\Ld$ of labeled ``reversed'' edges called the {\em suffix links} of $\DAWG(y)$
 by
  \[
   \Ld = \{(\Reqc{ax}, a, \Reqc{x}) \mid x, ax \in \Substr(y), a \in \Sigma, \Reqc{ax} \neq \Reqc{x} \}.
  \]
\end{definition}

See Figure~\ref{fig:trie_tree_DAWG}
for examples of $\STrie(y)$, $\ST(y)$, $\DAWG(y)$, and $\CDAWG(y)$.

\begin{theorem}[\cite{Blumer85}]
For any string $y$ of length $n > 2$,
the number of nodes in $\DAWG(y)$ is at most $2n-1$
and the number of edges in $\DAWG(y)$ is at most $3n - 4$.
\end{theorem}

Minimization of $\STrie(y)$ to $\DAWG(y)$ can be done by
merging isomorphic subtrees of $\STrie(y)$ which are rooted
at nodes connected by a chain of suffix links of $\STrie(y)$.
Since the substrings represented by these merged nodes
end at the same positions in $y$,
each node of $\DAWG(y)$ forms an equivalence class $\Reqc{x}$.
We will make an extensive use of this property
in our $O(n)$-time construction algorithm
for $\DAWG(y)$ over an integer alphabet.

\begin{figure}[tb]
 \centerline{\includegraphics[scale=0.45]{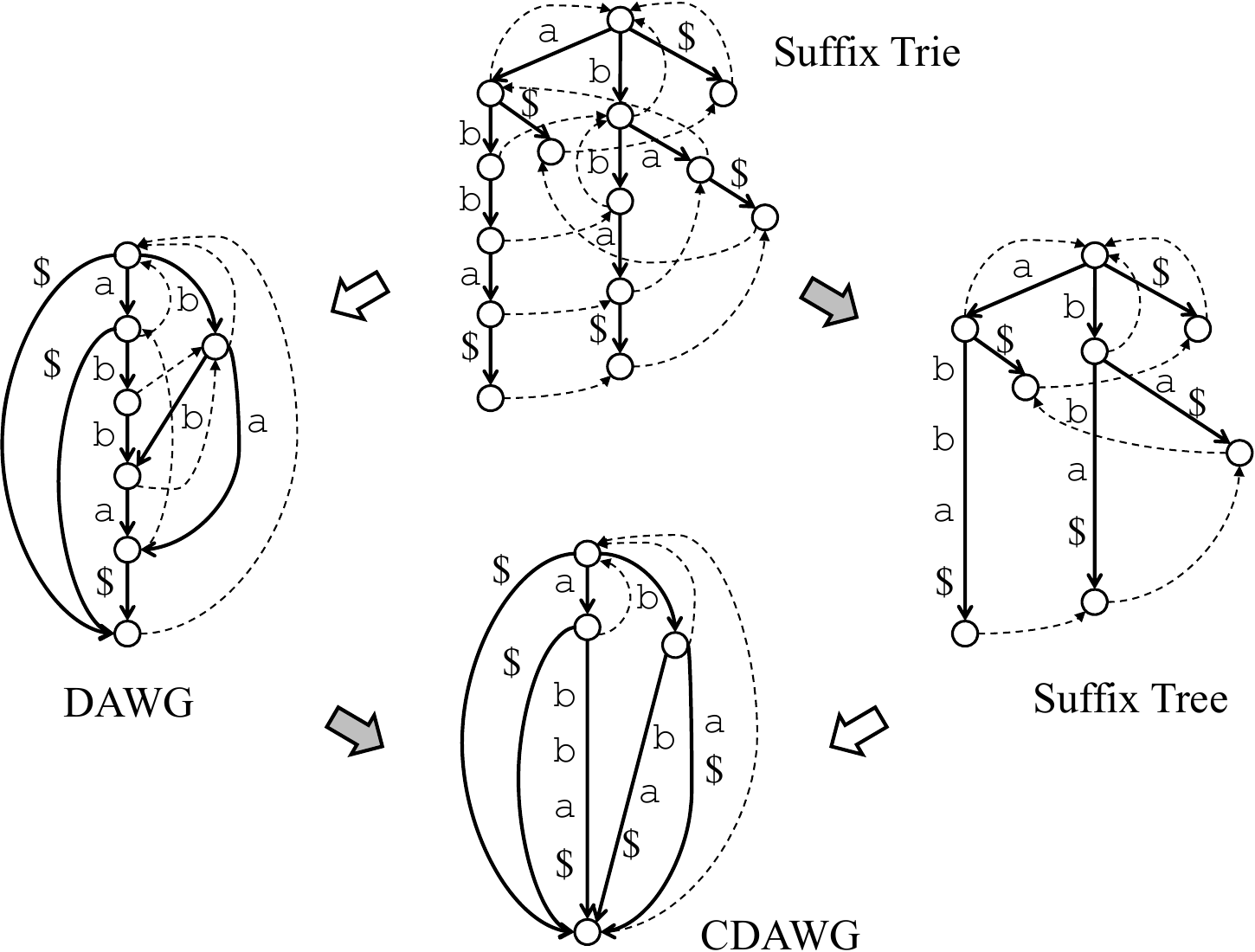}}
 \caption{The suffix trie, the suffix tree, the DAWG, and the CDAWG
  for string $y=\mathtt{abba\$}$. The solid arcs represent edges,
  and the broken arcs represent suffix links.
  The DAWG can be obtained by merging isomorphic subtrees of the suffix trie,
  while the suffix tree can be obtained by performing path-compressions on the suffix trie.
  The CDAWG can be obtained by merging isomorphic subtrees of the suffix tree,
  or by performing path-compressions on the DAWG.
  }
 \label{fig:trie_tree_DAWG}
\end{figure}

\subsection{Minimal Absent Words}

A string $x$ is said to be an {\em absent word} of another string $y$
if $x \notin \Substr(y)$.
An absent word $x$ of $y$ is said to be
a {\em minimal absent word} (\emph{MAW}) of $y$
if $\Substr(x) \setminus \{x\} \subset  \Substr(y)$.
The set of all MAWs of $y$ is denoted by $\MAW(y)$.
For example, if
$\Sigma = \{\mathtt{a}, \mathtt{b}, \mathtt{c}\}$
and $y = \mathtt{abaab}$, then $\MAW(y) = \{\mathtt{aaa}, \mathtt{aaba}, \mathtt{bab}, \mathtt{bb}, \mathtt{c}\}$.

\begin{lemma}[\cite{MignosiRS02}] \label{lem:MAW_bounds}
  For any string $y \in \Sigma^*$,
  $\sigma \leq |\MAW(y)| \leq (\sigma_y - 1)(|y|-1) + \sigma$,
  where $\sigma = |\Sigma|$ and $\sigma_y$ is the number of distinct
  characters occurring in $y$.
  These bounds are tight.
\end{lemma}

The next lemma follows from the definition of MAWs.
\begin{lemma} \label{lem:maw}
  Let $y$ be any string.
  For any characters $a, b \in \Sigma$
  and string $x \in \Sigma^*$,
  $axb \in \MAW(y)$ iff $axb \notin \Substr(y)$,
  $ax \in \Substr(y)$, and $xb \in \Substr(y)$.
\end{lemma}
By Lemma~\ref{lem:maw},
we can encode each MAW $axb$ of $y$ in $O(1)$ space by $(i, j, b)$,
where $ax = y[i..j]$.


\section{Folklore Algorithm to Construct $\DAWG(\rev{y})$ from $\ST(y)$ in $O(n)$ Time for Integer Alphabet}
\label{sec:folklore_algo}

For the suffix tree $\ST(y)=(\Vs,\Es)$,
we define the set $\WL$ of edges called \emph{Weiner links} of $\ST(y)$ by
$$\WL = \{(x, a, ax\beta) \mid  a \in \Sigma, \beta \in \Sigma^*, x \in \Vs, \Lrep{ax} = ax\beta \in \Vs\}.$$
The following lemma on the relation between Weiner links and DAWG is well known.

\begin{lemma}[\cite{Blumer85,Chen85}]
\label{lem:revDAWG}
For the suffix tree $\ST(y)=(\Vs,\Es)$ of a string $y$,
$(\mathrm{V}, \mathrm{E})$ is the DAWG of $\rev{y}$,
where $\mathrm{V} = \{ \Leqc{x} \mid x \in \Vs \}$ and
$\mathrm{E} = \{ (\Leqc{x}, a, \Leqc{ax\beta}) \mid (x, a, ax\beta) \in \WL \}$,
and $\mathrm{L} = \{(\Leqc{xb\beta}, b, \Leqc{x}) \mid b\in \Sigma, x,\beta \in \Sigma^*, (x, b\beta, xb\beta) \in \Es \}$ are the suffix links of the DAWG $(\mathrm{V}, \mathrm{E})$ of $\rev{y}$.
\end{lemma}

Each Weiner-link $(x, a, ax\beta) \in \WL$ is called \emph{explicit}
\sinote*{changed}{%
  if $\beta = \varepsilon$ (namely, $ax \in \Vs$ or equivalently $\Lrep{ax} = ax$),
}%
and it is called \emph{implicit} otherwise.
By definition, the explicit Weiner links are identical to the reversed suffix links.
Figure~\ref{fig:STwithSLWL} shows a concrete example of suffix links and Weiner links.

\begin{figure}[tb]
  \centerline{
  \includegraphics[scale=0.45,clip]{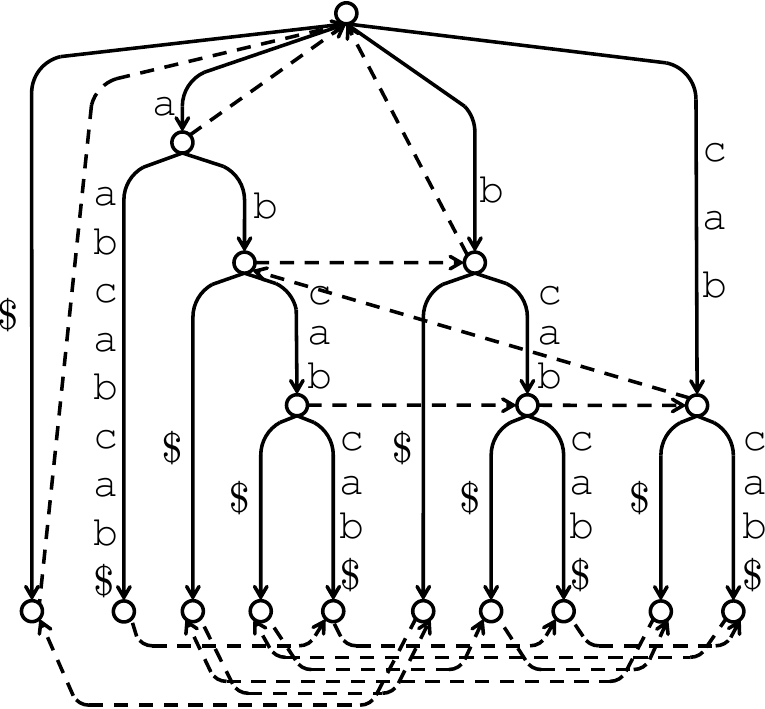}
  \hfill
  \includegraphics[scale=0.45,clip]{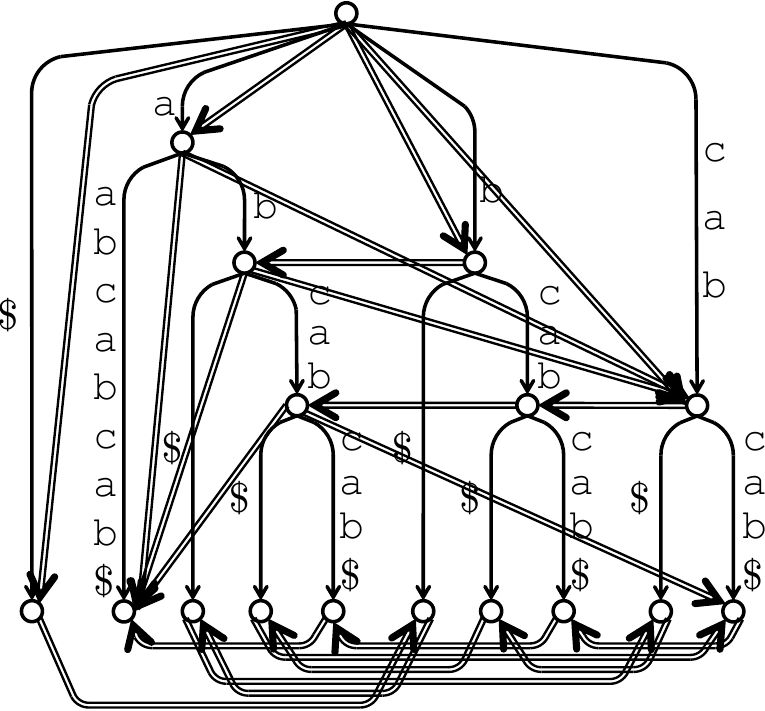}
  }
  \caption{
  (Left): Illustration of the suffix tree $\ST(y)$ for $y = \mathtt{aabcabcab}\$$ with its suffix links.
  (Right): Illustration of the suffix tree $\ST(y)$ for $y$ with its Weiner links.
  }
  \label{fig:STwithSLWL}
\end{figure}

By the above argument, what is left is to compute the implicit Weiner links.
In this section, we describe a folklore algorithm for computing the Weiner links
in linear time, from a given suffix tree $\ST(y)$ augmented with the suffix links.

Let $\parent^1(x)$ be $\parent(x)$ and $\parent^i(x)$ be the parent of $\parent^{i-1}(x)$ for $i > 1$ in the suffix tree $\ST(y)$.
Then the following properties hold.
\begin{observation}[\cite{DBLP:journals/corr/abs-1302-3347}]
\label{obs:WeinerLink}
For each implicit Weiner link $(x, a, ax\beta) \in \WL$ of $\ST(y)$ with $\beta \in \Sigma^+$,
there exists the explicit Weiner link $(x\beta, a, ax\beta)$.
For each explicit Weiner link $(w,a,aw)$,
if $k$ is the smallest integer such that $\parent(aw) = a \parent^k(w)$
(i.e. $k$ is the smallest integer such that $(\parent^k(w), a, \parent(aw))$ is an explicit Weiner link),
then $(\parent^i(w), a, aw)$ is an implicit Weiner link for each $1 \leq i < k$.
See Figure~\ref{fig:WLConst} and Figure~\ref{fig:WL_concrete}.
\end{observation}
From Observation~\ref{obs:WeinerLink},
every implicit Weiner link of $\ST(y)$ is of the form $(\parent^i(w), a, aw)$ described above
for some explicit Weiner link $(w,a,aw)$.
Thus we can obtain all implicit Weiner links in linear time by computing $(\parent^i(w), a, aw)$ for
all explicit Weiner links $(w,a,aw)$ which are the reversed edges of $\Ls$ of $\ST(y)$.
This computation takes $O(n)$ total time.

\begin{figure}
  \centerline{
    \includegraphics[scale=0.6]{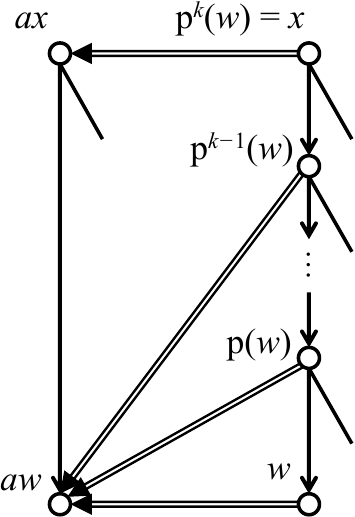}
    \hspace{80pt}
    \includegraphics[scale=0.6]{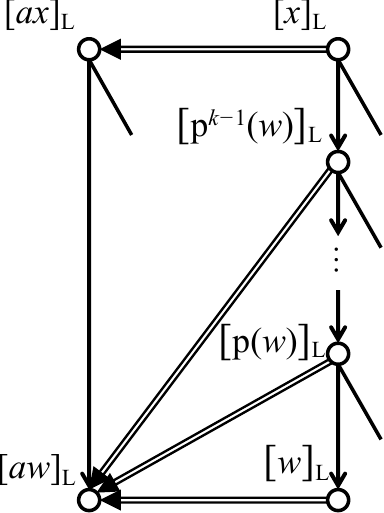}
  }
  \caption{
  (Left) Illustration for implicit Weiner links and explicit Weiner links.
    The double-lined arrows represent Weiner links,
    the single-lined arrows represent suffix tree edges,
    and the white circles represent suffix tree nodes.
    (Right) Illustration for a relation between implicit Weiner links and
    their corresponding equivalence classes w.r.t. $\Leqr$.
  }
  \label{fig:WLConst}
\end{figure}

\begin{figure}
  \begin{center}
    \includegraphics[width=0.9\textwidth]{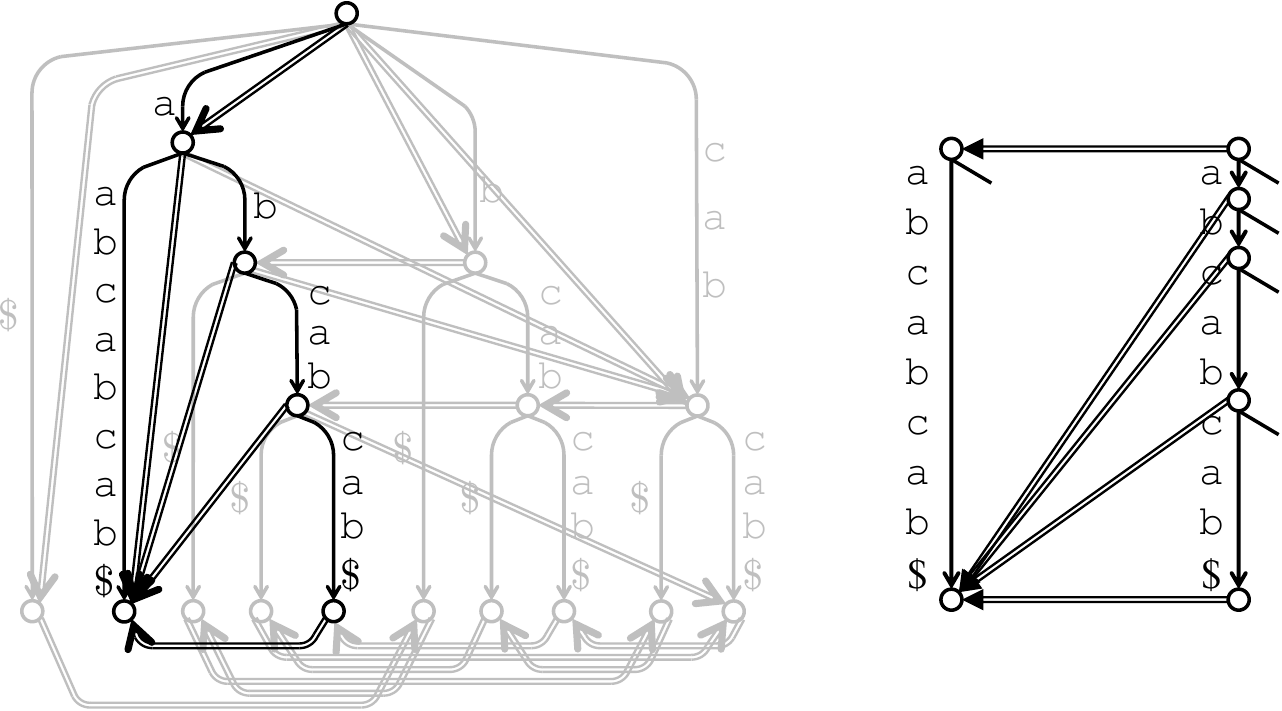}
  \end{center}
  \caption{
  Illustration of some Weiner links in $ \ST(\mathtt{aabcabcab}\$)$
  }
  \label{fig:WL_concrete}
\end{figure}

\begin{theorem}\label{theo:WL_const}
Given $\ST(y)$ and its suffix links $\Ls$ of string $y$ of length $n$ over an integer alphabet,
the Weiner links $\WL$ of $\ST(y)$ can be computed in $O(n)$ time.
\end{theorem}

We can obtain $\DAWG(\rev{y})$ from $\ST(y)$ by Lemma~\ref{lem:revDAWG} and Theorem~\ref{theo:WL_const}.
We remark however that the edges of $\DAWG(\rev{y})$
might not be sorted, since the Weiner links $\WL$ of $\ST(y)$ might not be sorted.
Still, we can easily sort all the edges of $\DAWG(\rev{y})$ in $O(n)$ total time
after they are constructed:
First, extract all edges of $\DAWG(\rev{y})$ by a standard traversal on
$\DAWG(\rev{y})$, which takes $O(n)$ time.
Next, radix sort them by their labels,
which takes $O(n)$ time because we assumed an integer alphabet
of polynomial size in $n$.
Finally, re-insert the edges to their respective nodes
in the sorted order.

\begin{theorem}
Given $\ST(y)$ and its suffix links $\Ls$ of string $y$ of length $n$ over an integer alphabet,
the edge-sorted $\DAWG(\rev{y})$ of the reversed string $\rev{y}$ with suffix links can be constructed in $O(n)$ time and space.
\end{theorem}

In some applications such as bidirectional pattern searches,
it is preferable that the in-coming suffix links at each node of $\DAWG(\rev{y})$
are also sorted in lexicographical order,
but the algorithm described above does not sort the suffix links.
However, one can sort the suffix links in $O(n)$ time
by the same technique applied to the edges of $\DAWG(\rev{y})$.

\begin{remark*}
  It is noteworthy that the parallel DAWG construction algorithm
  proposed by Breslauer and Hariharan~\cite{DBLP:journals/ppl/BreslauerH96}
  can be seen as a parallel version of the above folklore algorithm.
\end{remark*}

\section{Constructing $\DAWG(y)$ from $\ST(y)$ in $O(n)$ Time for Integer Alphabet}
\label{sec:DAWG_linear_const}

In this section, we present an optimal $O(n)$-time algorithm
to construct $\DAWG(y)$ with suffix links $\Ld$
for a given string $y$ of length $n$ over an integer alphabet.
Our algorithm constructs $\DAWG(y)$ with suffix links $\Ld$
from $\ST(y)$ with suffix links $\Ls$.
The following result is known.

\begin{theorem}[\cite{Farach-ColtonFM00}] \label{theo:ST_const}
Given a string $y$ of length $n$ over an integer alphabet,
edge-sorted $\ST(y)$ with suffix links $\Ls$ can be computed in $O(n)$ time.
\end{theorem}


\begin{figure}
  \begin{center}
    \includegraphics[width=0.5\textwidth]{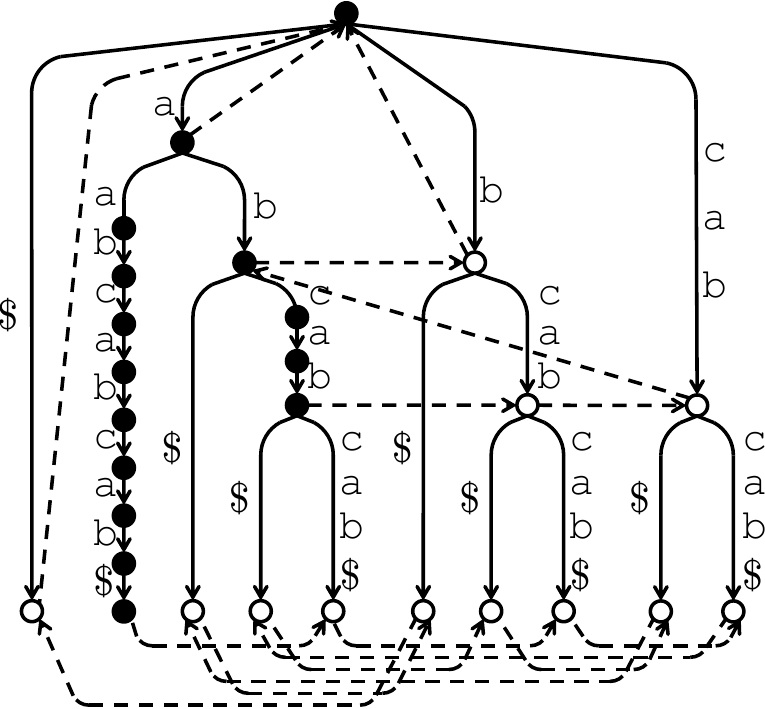}
  \end{center}
  \caption{Illustration of $\AST(y)$ with string $y = \mathtt{aabcabcab\$}$.}
  \label{fig:stdawg}
\end{figure}

Let $\mathcal{L}$ and $\mathcal{R}$
be, respectively, the sets of longest elements of all equivalence classes
on $y$ w.r.t. $\Leqr$ and $\Reqr$, namely,
$\mathcal{L} = \{\Lrep{x} \mid x \in \Substr(y)\}$ and
$\mathcal{R} = \{\Rrep{x} \mid x \in \Substr(y)\}$.
Let $\AST(y) = (\AVs, \AEs)$ be the edge-labeled rooted tree
obtained by adding extra nodes for strings in $\mathcal{R}$ to $\ST(y)$, namely,
\begin{eqnarray*}
  \AVs & = & \{x \mid x \in \mathcal{L} \cup \mathcal{R} \}, \\
  \AEs & = & \{(x, \beta, x\beta) \mid x,x\beta \in \AVs, \beta \in \Sigma^+,
  1 \leq \forall i < |\beta|,x\cdot\beta[1..i]\notin \AVs\}.
\end{eqnarray*}
Notice that the size of $\AST(y)$ is $O(n)$,
since $|\mathcal{L} \cup \mathcal{R}| \leq |\Vs| + |\Vd| = O(n)$,
where $\Vs$ and $\Vd$ are respectively the sets of nodes of
$\ST(y)$ and $\DAWG(y)$.

A node $x \in \AVs$ of $\AST(y)$ is called \emph{black}
iff $x \in \mathcal{R}$.
See Figure~\ref{fig:stdawg} for an example of $\AST(y)$.

\begin{lemma} \label{lem:black_nodes_monotonicity}
For any $x \in \Substr(y)$,
if $x$ is represented by a black node in $\AST(y)$,
then every prefix of $x$ is also represented by a black node in $\AST(y)$.
\end{lemma}

\begin{proof}
Since $x$ is a black node, $x = \Rrep{x}$.
Assume on the contrary that there is a proper prefix $z$ of $x$
such that $z$ is not represented by a black node.
Let $zu = x$ with $u \in \Sigma^+$.
Since $z \Reqr \Rrep{z}$, we have $x = zu \Reqr \Rrep{z}u$.
On the other hand, since $z$ is not black, we have $|\Rrep{z}| > |z|$.
However, this contradicts that $x$ is the longest member $\Rrep{x}$ of $\Reqc{x}$.
Thus, every prefix of $x$ is also represented by a black node.
\end{proof}

\begin{lemma} \label{lem:black_leaf_condition}
  For any string $y$,
  let $\BT(y)$ be the trie consisting only of the
  black nodes of $\AST(y)$.
  Then, every leaf $\ell$ of $\BT(y)$ is a node of the original suffix tree
  $\ST(y)$.
\end{lemma}

\begin{proof}
  Assume on the contrary that
  some leaf $\ell$ of $\BT(y)$ corresponds to
  an internal node of $\AST(y)$ that has exactly one child.
  Since any substring in $\mathcal{L}$ is represented
  by a node of the original suffix tree $\ST(y)$,
  we have $\ell \in \mathcal{R}$.
  Since $\ell = \Rrep{\ell}$,
  $\ell$ is the longest substring of $y$
  which has ending positions $\EndPos(\ell)$ in $y$.
  This implies one of the following situations:
  (1) occurrences of $\ell$ in $y$ are immediately preceded
  by at least two distinct characters $a \neq b$,
  (2) $\ell$ occurs as a prefix of $y$ and
  all the other occurrences of $\ell$ in $y$ are immediately preceded by
  a unique character $a$, or
  (3) $\ell$ occurs exactly once in $y$ as its prefix.
  Let $u$ be the only child of $\ell$ in $\AST(y)$,
  and let $\ell z = u$, where $z \in \Sigma^+$.
  By the definition of $\ell$, $u$ is not black.
  On the other hand, in any of the situations (1)-(3),
  $u = \ell z$ is the longest substring of $y$ which has ending positions
  $\EndPos(u)$ in $y$.
  Hence we have $u = \Rrep{u}$ and $u$ must be black, a contradiction.
  Thus, every leaf $\ell$ of $\BT(y)$ is a node of the original suffix tree
  $\ST(y)$.
\end{proof}

\begin{lemma}[\cite{Narisawa07}] \label{lem:original_node_black_condition}
For any node $x \in \Vs$ of the original suffix tree $\ST(y)$,
its corresponding node in $\AST(y)$ is black iff
(1) $x$ is a leaf of the suffix link tree $\SLT(y)$, or
(2) $x$ is an internal node of $\SLT(y)$ and
for any character $a \in \Sigma$ such that $ax \in \Vs$,
$|\BegPos(ax)| \neq |\BegPos(x)|$.
\end{lemma}

Using Lemma~\ref{lem:black_leaf_condition} and
Lemma~\ref{lem:original_node_black_condition},
we can compute all leaves of $\BT(y)$ in $O(n)$ time
by a standard traversal on the suffix link tree $\SLT(y)$.
Then, we can compute all internal black nodes of $\BT(y)$
in $O(n)$ time using Lemma~\ref{lem:black_nodes_monotonicity}.
Now, by Theorem~\ref{theo:ST_const}, the next lemma holds:

\begin{lemma} \label{lem:AugmentedST_const}
Given a string $y$ of length $n$ over an integer alphabet,
edge-sorted $\AST(y)$ can be constructed in $O(n)$ time.
\end{lemma}

We construct $\DAWG(y)$ with suffix links $\Ld$ from $\AST(y)$,
as follows.
First, we construct a DAG $D$, which is initially
equivalent to the trie $\BT(y)$
consisting only of the black nodes of $\AST(y)$.
Our algorithm adds edges and suffix links to $D$,
so that the DAG $D$ will finally become $\DAWG(y)$.
In so doing, we traverse $\AST(y)$ in post order.
For each black node $x$ of $\AST(y)$ visited in the post-order traversal,
which is either an internal node or a leaf of the original suffix tree $\ST(y)$,
we perform the following:
Let $\parent(x)$ be the parent of $x$ in the \emph{original suffix tree} $\ST(y)$.
It follows from Lemma~\ref{lem:black_nodes_monotonicity} that
every prefix $x'$ of $x$ with $|\parent(x)| \leq |x'| \leq |x|$
is represented by a black node.
For each black node $x'$ in the path from $\parent(x)$ to $x$ in the DAG $D$,
we compute the in-coming edges to $x'$ and the suffix link of $x'$.

\begin{figure}[t!]
  \centerline{
    \includegraphics[scale=0.35,clip]{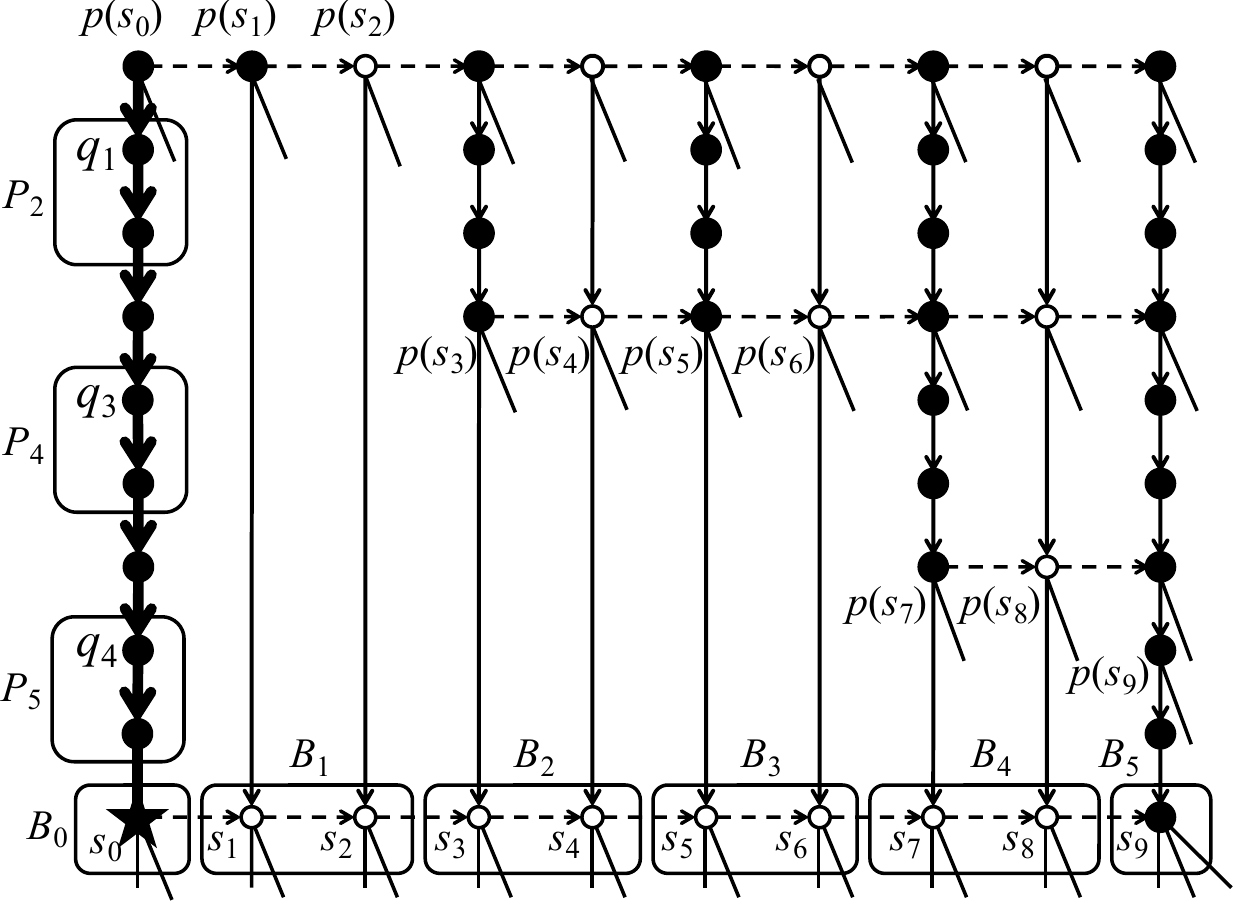}
    \hfil
    \includegraphics[scale=0.35,clip]{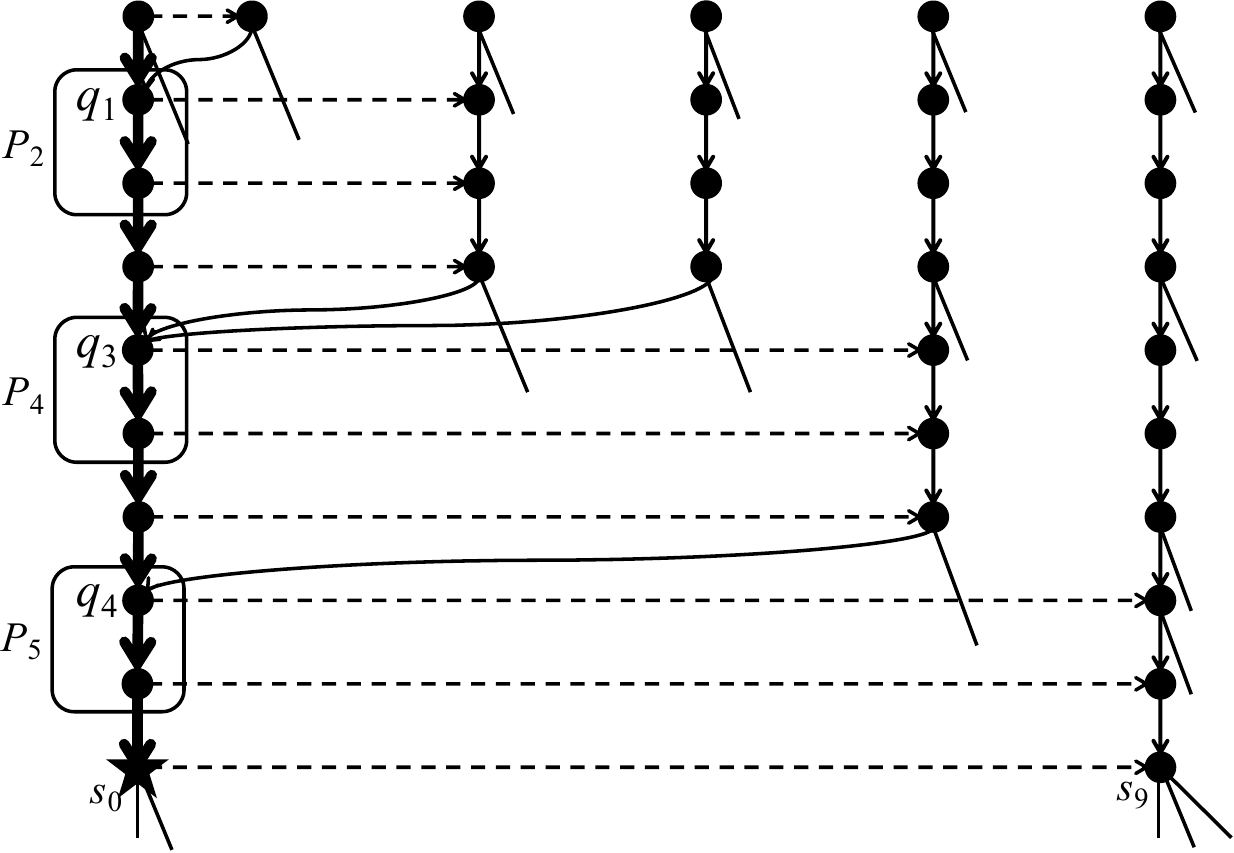}
  }
\caption{
(Left): Illustration for a part of $\AST(y)$,
where the branching nodes are those
that exist also in the original suffix tree $\ST(y)$.
Suppose we have just visited node $x = s_0$ (marked by a star)
in the post-order traversal on $\AST(y)$.
Here, $s_0, \ldots, s_9$ are connected by a chain of suffix links
starting from $s_0$, and $s_9$ is the first black node after $s_0$ in the chain.
In the corresponding DAG $D$,
we will add in-coming edges to the black nodes
in the path from $\parent(x)$ to $x$,
and will add suffix links from these black nodes in the path.
The sequence $s_0, \ldots, s_m$ of nodes in $\AST(y)$ is partitioned into blocks,
such that the parents of the nodes in the same block belong to the
same equivalence class w.r.t. $\Reqr$.
(Right): The in-coming edges and the suffix links have been added
to the nodes in the path from $\parent(x)$ to $x = s_0$.
}
\label{fig:suffix_link_const}
\end{figure}

Let $s_0, \ldots, s_m$ be the sequence of nodes
connected by a chain of suffix links starting from $s_0 = x$,
such that $|\BegPos(s_i)| = |\BegPos(s_0)|$ for all $0 \leq i \leq m - 1$
and $|\BegPos(s_{m})| > |\BegPos(s_0)|$
(see the left diagram of Figure~\ref{fig:suffix_link_const}).
In other words, $s_{m}$ is
the first black node after $s_0$ in the chain of suffix links
(this is true by Lemma~\ref{lem:original_node_black_condition}).
Since $|s_{i}| = |s_{i-1}|-1$ for every $1 \leq i \leq m - 1$,
$\EndPos(s_i) = \EndPos(s_0)$.
Thus, $s_0, \ldots, s_{m-1}$ form a single equivalence class w.r.t. $\Reqr$
and are represented by the same node as $x = s_0$ in the DAWG.

For any $0 \leq i \leq m - 1$, let $d(s_i) = |s_i| - |\parent(s_i)|$.
Observe that the sequence $d(s_0), \ldots, d(s_{m})$
is monotonically non-increasing by the property of suffix trees.
We partition the sequence $s_0, \ldots, s_m$ of nodes
into blocks so that the parents of all nodes in the same block
belong to the same equivalence class w.r.t. $\Reqr$.
Let $r$ be the number of such blocks,
and for each $0 \leq k \leq r-1$,
let $B_k = s_{i_k}, \ldots, s_{i_{k+1}-1}$ be the $k$th such block.
We can easily compute all these blocks by comparing $|\parent(s_{i-1})|$
and $|\parent(s_i)|$ for each pair $s_{i-1}$ and $s_{i}$ of consecutive elements in the sequence $s_1,\ldots, s_m$ of nodes.
Note that for each block $B_k$,
$\parent(s_{i_k})$ is the only black node among the parents
$\parent(s_{i_k}), \ldots, \parent(s_{i_{k+1}-1})$ of the nodes in $B_k$,
since it is the longest one in its equivalence class $\Reqc{\parent(s_{i_k})}$.
Also, every node in the same block has the same value for function $d$.
Thus, for each block $B_k$,
we add a new edge $(\parent(s_{i_k}), b_k, q_k)$ to the DAG $D$,
where $q_k$ is the (black) ancestor of $x$ such that
$|q_k| = |x| - d(s_{i_k}) + 1$,
and $b_k$ is the first character of the label of the edge
from $\parent(s_{i_k})$ to $s_{i_k}$ in $\AST(y)$.
Notice that this new edge added to $D$ corresponds to the edges
between the nodes in the block $B_k$ and their parents in $\AST(y)$.
We also add a suffix link $(\parent(q_k), a, \parent(s_{i_k}))$ to $D$,
where $a = s_{i_k-1}[1]$.
See also the right diagram of Figure~\ref{fig:suffix_link_const}.

For each $2 \leq k \leq r-1$,
let $P_k$ be the path from $q_{k-1}$ to $g_k$,
where $g_k = \parent(\parent(q_{k}))$ for $2 \leq k \leq r-2$
and $g_{r-1} = x = s_0$.
Each $P_k$ is a sub-path of the path from $\parent(s_0)$ to $s_0$,
and every node in $P_k$ has not been given their suffix link yet.
For each node $v$ in $P_k$,
we add the suffix link from $v$ to the
ancestor $u$ of $s_{i_k}$ such that $|s_{i_k}| - |u| = |s_0| - |v|$.
See also the right diagram of Figure~\ref{fig:suffix_link_const}.

Repeating the above procedure for all black nodes of $\AST(y)$
that are either internal nodes or leaves of the original suffix tree $\ST(y)$
in post order,
the DAG $D$ finally becomes $\DAWG(y)$ with suffix links $\Ld$.
We remark however that the edges of $\DAWG(y)$
might not be sorted, since the edges that exist in $\AST(y)$
were firstly inserted to the DAG $D$.
Still, we can easily sort all the edges and suffix links of $\DAWG(y)$ in $O(n)$ total time
after they are constructed by the same technique as in Section~\ref{sec:folklore_algo}.

\begin{theorem} \label{theo:DAWG_linear_const}
Given a string $y$ of length $n$ over an integer alphabet,
we can compute edge-sorted $\DAWG(y)$ with suffix links $\Ld$
in $O(n)$ time and space.
\end{theorem}

\begin{proof}
The correctness can easily be seen if one recalls
that minimizing $\STrie(y)$ based on its suffix links produces $\DAWG(y)$.
The proposed algorithm simulates this minimization
using only the subset of the nodes of $\STrie(y)$
that exist in $\AST(y)$.
The out-edges of each node of $\DAWG(y)$ are sorted in lexicographical order
as previously described.

We analyze the time complexity of our algorithm.
We can compute $\AST(y)$ in $O(n)$ time by Lemma~\ref{lem:AugmentedST_const}.
The initial trie for $D$ can easily be computed in $O(n)$ time
from $\AST(y)$.
Let $x$ be any node visited in the post-order traversal on $\AST(y)$
that is either an internal node or a leaf of the original suffix tree $\ST(y)$.
The cost of adding the new in-coming edges
to the black nodes in the path from $\parent(x)$ to $x = s_0$
is linear in the number of nodes in the sequence $s_0, \ldots, s_m$
connected by the chain of suffix links starting from $s_0 = x$.
Since $s_0$ and $s_m$ are the only black nodes in the sequence,
it follows from Lemma~\ref{lem:original_node_black_condition} that
the chain of suffix links from $s_0$ to $s_m$ is a non-branching path
of the suffix link tree $\SLT(y)$.
This implies that the suffix links in this chain
are used only for node $x$ during the post-order traversal of $\AST(y)$.
Since the number of edges in $\SLT(y)$ is $O(n)$,
the amortized cost of adding each edge to $D$ is constant.
Also, the total cost to sort all edges is $O(n)$,
as was previously explained.
Now let us consider the cost of adding
the suffix links from the nodes in each sub-path $P_k$.
For each node $v$ in $P_k$,
the destination node $v$ can be found in constant time
by simply climbing up the path from $s_{i_k}$ in the chain of suffix links.
Overall, the total time cost to transform the trie for $D$
to $\DAWG(y)$ is $O(n)$.

The working space is clearly $O(n)$.
\end{proof}

\begin{example}
  Figure~\ref{fig:snapshots_DAWG_const} shows snapshots of
  the DAWG construction for string $y=\mathtt{aabcabcab}\$$ by our algorithm.
  Step 0: (Left): $\AST(y)$ with suffix links $\Ls$ and (Right): the initial trie for $D$. We traverse $\AST(y)$ in post order.
  Step 1: We arrived at black leaf node $x_1 = \mathtt{aabcabcab}\$$ (indicated by a star). We determine the in-coming edges and suffix links for the black nodes in the path from $\parent(x_1) = \mathtt{a}$ and $x_1$ (indicated by thick black lines). To the right is the resulting DAG $D$ for this step.
  Step 2: We arrived at black branching node $x_2 = \mathtt{abcab}$ (indicated by a star). We determine the in-coming edges and suffix links for the black nodes in the path from $\parent(x_2) = \mathtt{ab}$ and $x_2$ (indicated by thick black lines). To the right is the resulting DAG $D$ for this step.
  Step 3: We arrived at black branching node $x_3 = \mathtt{ab}$ (indicated by a star). We determine the in-coming edges and suffix links for the black nodes in the path from $\parent(x_3) = \mathtt{a}$ and $x_3$ (indicated by thick black lines). To the right is the resulting DAG $D$ for this step.
  Step 4: We arrived at black branching node $x_4 = \mathtt{a}$ (indicated by a star). We determine the in-coming edges and suffix links for the black nodes in the path from $\parent(x_4) = \varepsilon$ and $x_4$ (indicated by thick black lines). To the right is the resulting DAG $D$ for this step.
  Since all branching and leaf black nodes have been processed,
  the final DAG $D$ is $\DAWG(y)$ with suffix links.
\end{example}

\begin{figure}[tbp]
  \centerline{
    \includegraphics[scale=0.3,clip]{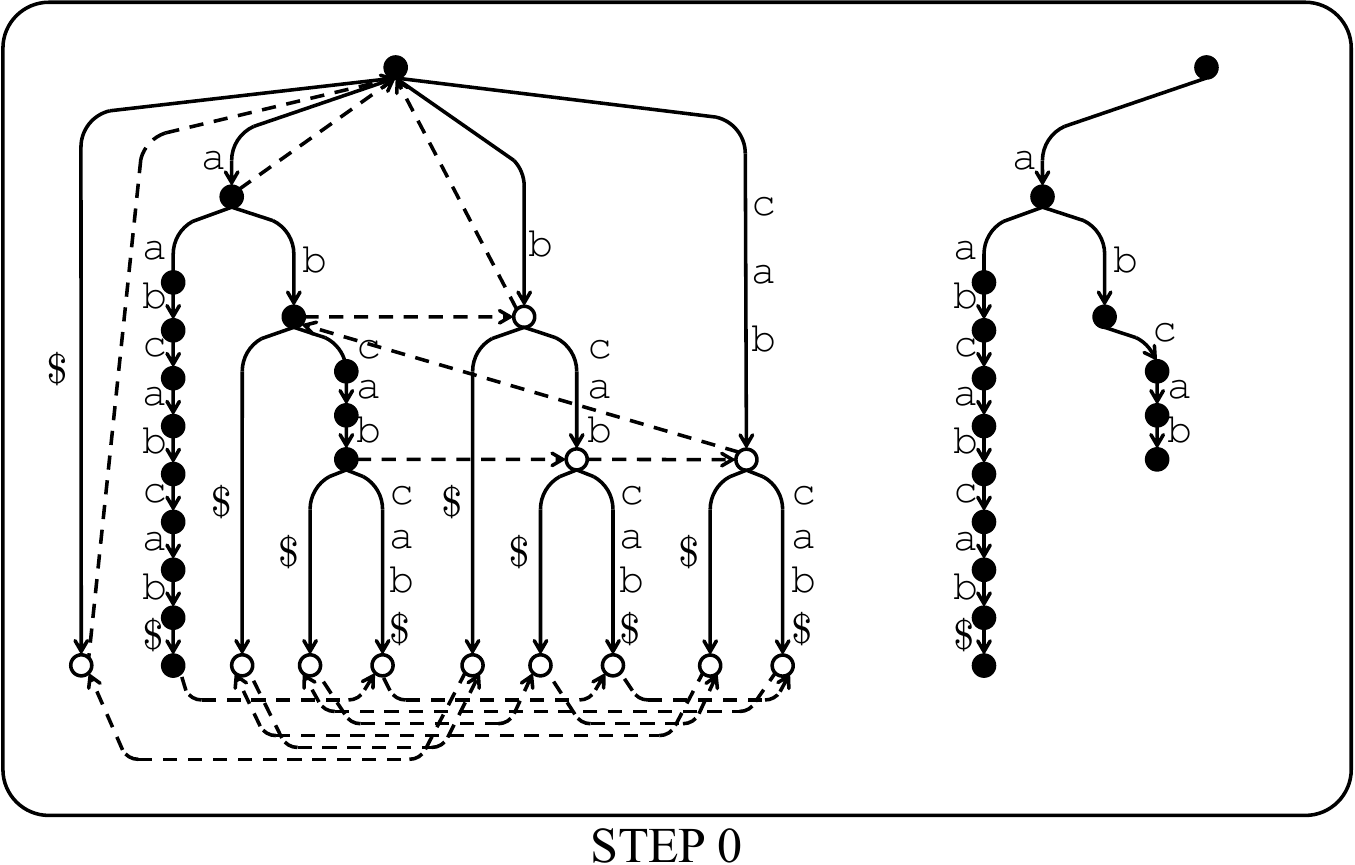}
    \hfill
    \includegraphics[scale=0.3,clip]{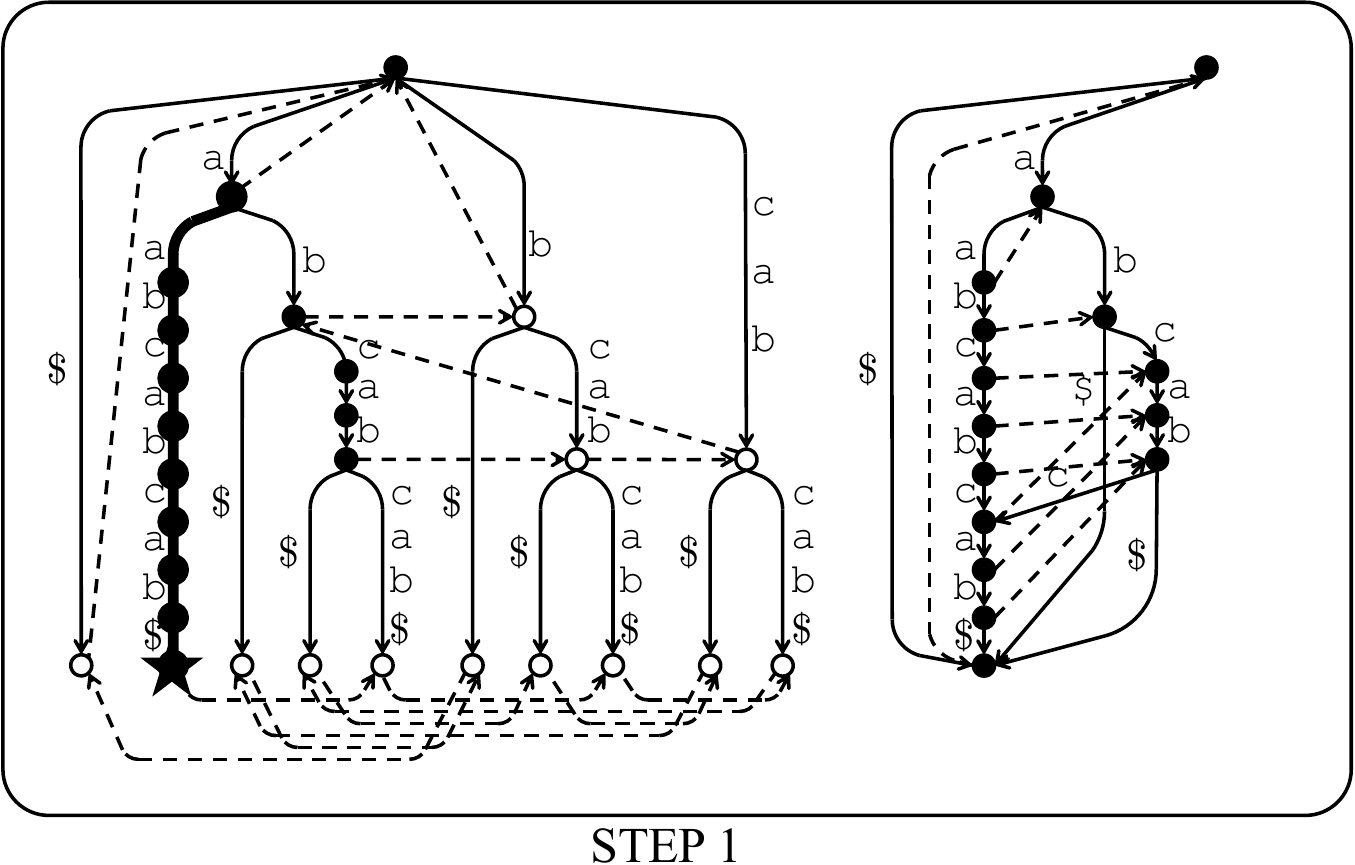}
  }
  \vspace*{1pc}
  \centerline{
    \includegraphics[scale=0.3,clip]{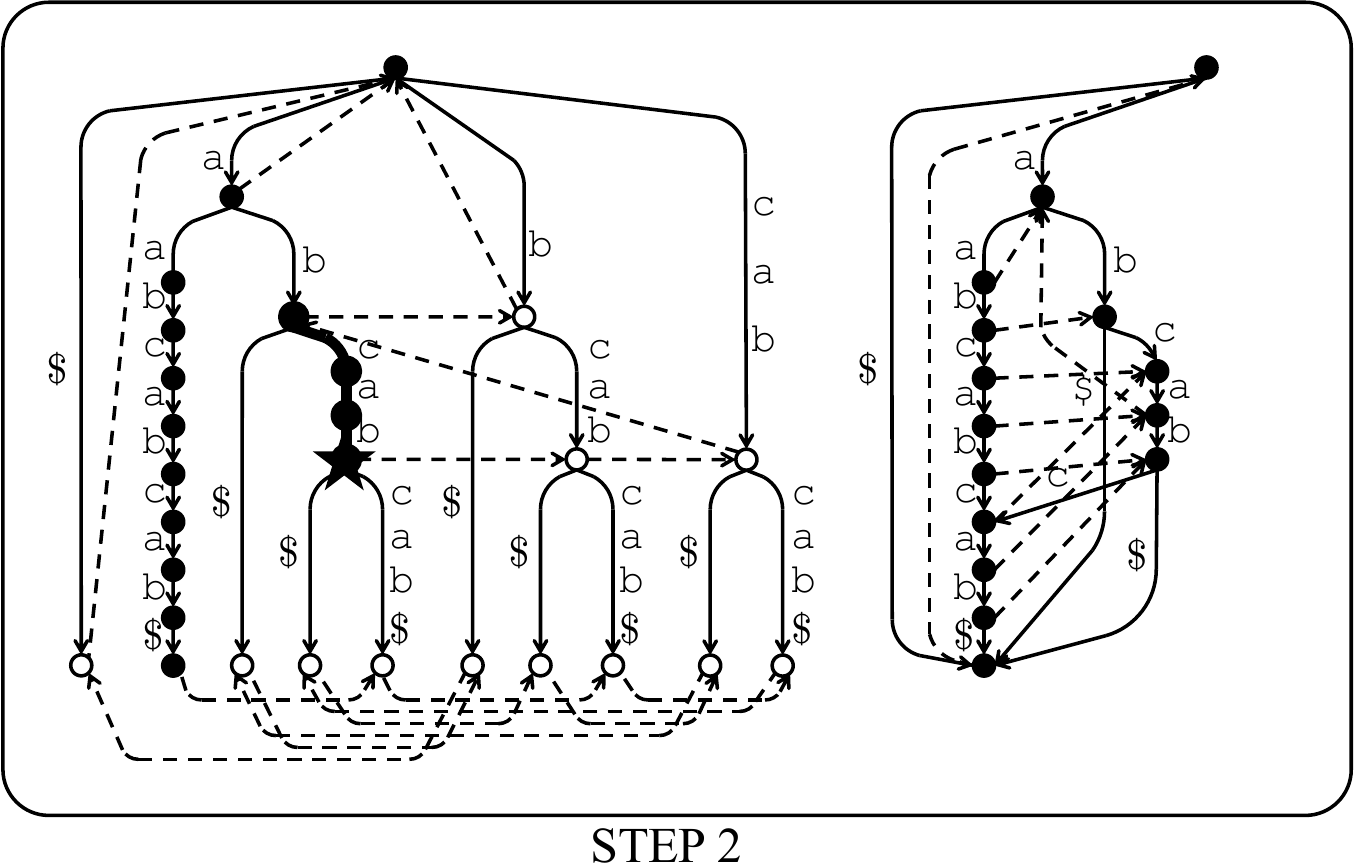}
    \hfill
    \includegraphics[scale=0.3,clip]{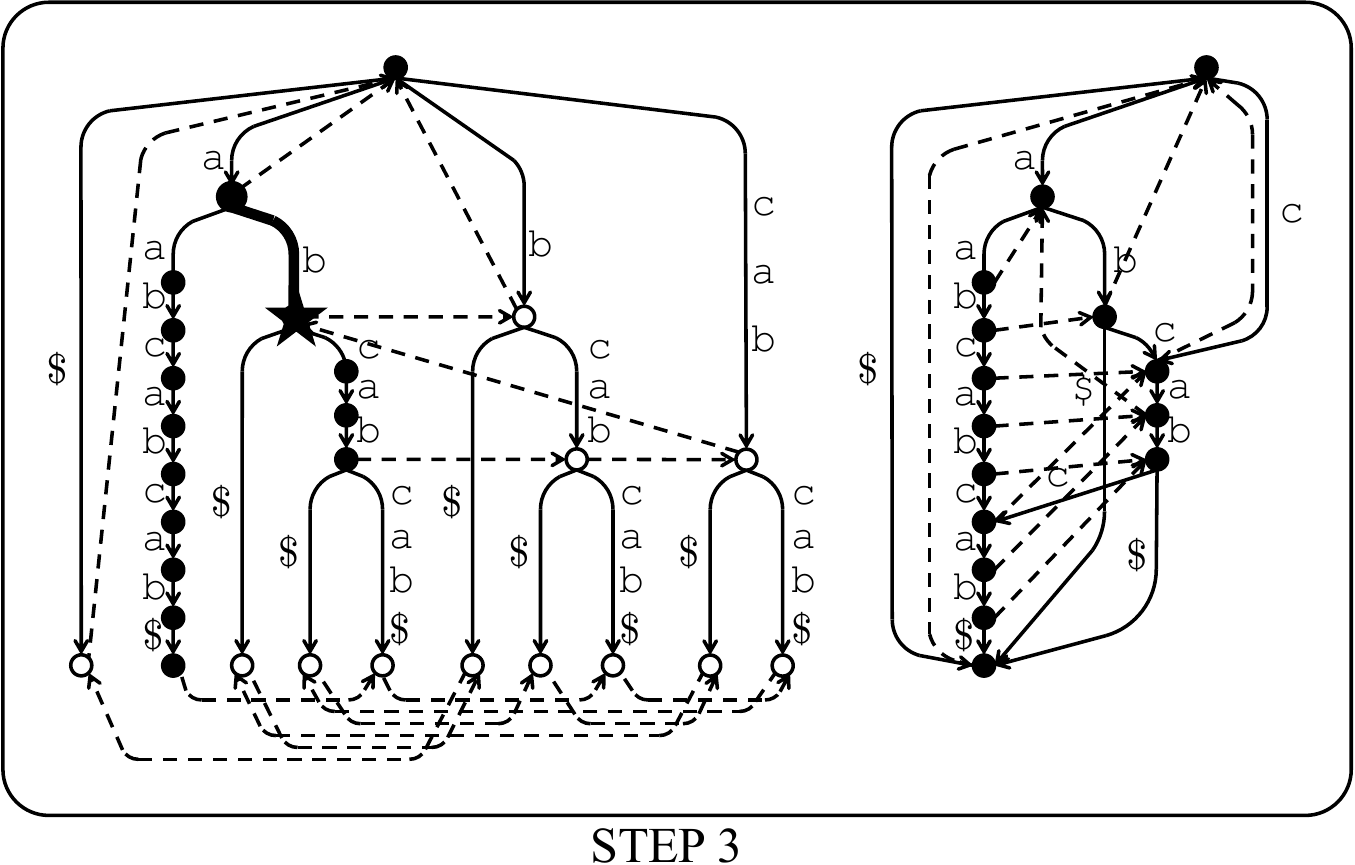}
  }
  \vspace*{1pc}
  \centerline{
    \includegraphics[scale=0.3,clip]{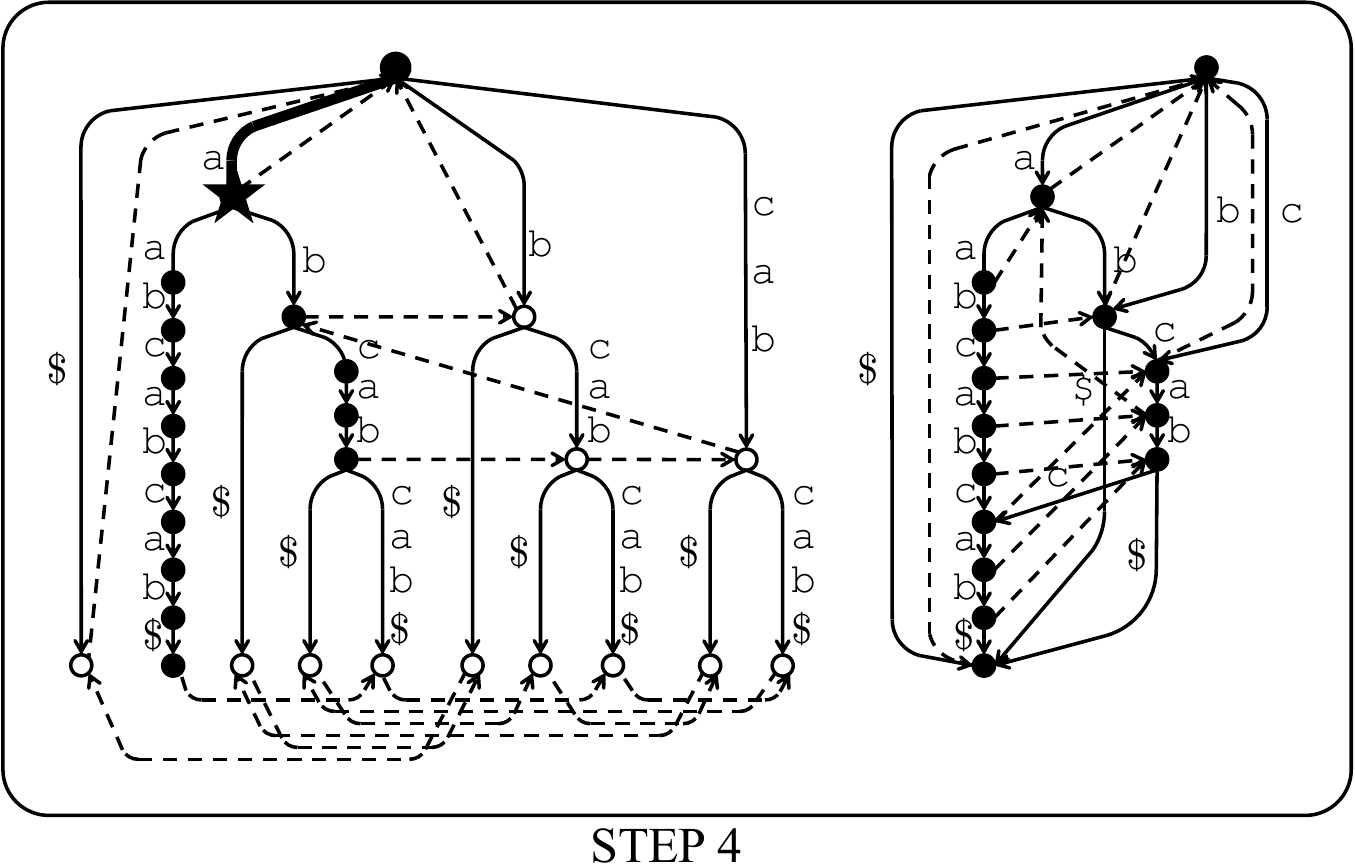}
  }
  \caption{Snapshots for the construction of $\DAWG(y)$ with $y = \mathtt{aabcabcab}\$$.
  }
  \label{fig:snapshots_DAWG_const}
\end{figure}

\section{Constructing Affix Trees in $O(n)$ Time for Integer Alphabet}
\label{sec:atree_linear_const}


\begin{figure}
  \begin{center}
    \includegraphics[width=0.45\textwidth]{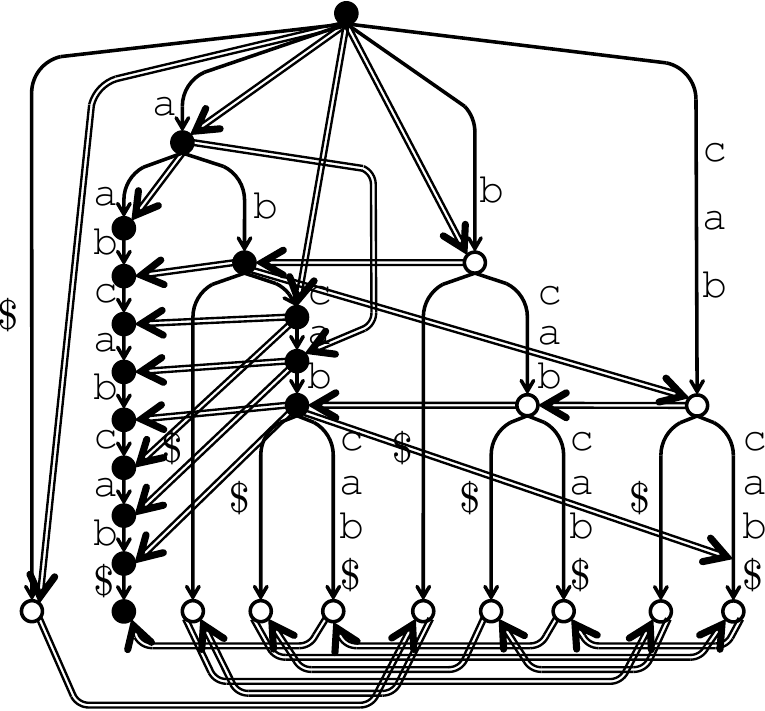}
  \end{center}
  \caption{Illustration of $\ATree(y)$ with string $y = \mathtt{aabcabcab\$}$.
    The solid arcs represent the forward edges in $\Eaf$,
    while the double-lined arcs represent the backward edges in $\Eab$.
    For simplicity, the labels of backward edges are omitted.}
  \label{fig:affix_tree}
\end{figure}

Let $y$ be the input string of length $n$ over an integer alphabet.
Recall the sets
$\mathcal{L} = \{\Lrep{x} \mid x \in \Substr(y)\}$ and
$\mathcal{R} = \{\Rrep{x} \mid x \in \Substr(y)\}$
introduced in Section~\ref{sec:DAWG_linear_const}.
For any set $S \subseteq \Sigma^* \times \Sigma^*$
of ordered pairs of strings,
let $S[1] = \{x_1 \mid (x_1, x_2) \in S \mbox{ for some } x_2 \in \Sigma^*\}$
and $S[2] = \{x_2 \mid (x_1, x_2) \in S \mbox{ for some } x_1 \in \Sigma^*\}$.

The affix tree~\cite{Stoye00} of string $y$, denoted $\ATree(y)$,
is a \emph{bidirectional} text indexing structure defined as follows:
\begin{definition}
  $\ATree(y)$ for string $y$ is an edge-labeled
  DAG $(\Va,\Ea) = (\Va, \Eaf \cup \Eab)$ which has
  two disjoint sets $\Eaf, \Eab$ of edges such that
  \begin{eqnarray*}
   \Va & = & \{(x, \rev{x}) \mid x \in \mathcal{L} \cup \mathcal{R} \},  \\
   \Eaf & = & \{((x, \rev{x}), \beta, (x\beta, \rev{\beta}\rev{x})) \mid x,x\beta \in \Va[1], \\
 && \phantom{\{} \beta \in \Sigma^+, 1\leq \forall i < |\beta|,x\cdot\beta[1..i]\notin \Va[1]\}, \\
   \Eab & = & \{((x, \rev{x}), \rev{\alpha}, (\alpha x, \rev{x}\rev{\alpha})) \mid \rev{x},\rev{x}\rev{\alpha} \in \Va[2], \\
 && \phantom{\{} \alpha \in \Sigma^+, 1\leq \forall i < |\alpha|,\rev{x}\cdot\rev{\alpha}[1..i]\notin \Va[2]\}.
  \end{eqnarray*}
  $\Eaf$ is the set of forward edges labeled by
  substrings of $y$,
  while $\Eab$ is the set of backward edges labeled by substrings of $\rev{y}$.
\end{definition}


\begin{theorem} \label{theo:affix_tree_linear_const}
Given a string $y$ of length $n$ over an integer alphabet of polynomial size in $n$, we can compute edge-sorted $\ATree(y)$
in $O(n)$ time and space.
\end{theorem}

\begin{proof}
  Clearly, there is a one-to-one correspondence between
  each node $(x, \rev{x}) \in \Va$ of $\ATree(y) = (\Va, \Eaf \cup \Eab)$ and
  each node $x \in \AVs$ of $\AST(y) = (\AVs, \AEs)$
  of Section~\ref{sec:DAWG_linear_const}
  (see also Figure~\ref{fig:stdawg} and Figure~\ref{fig:affix_tree}).
  Moreover, there is a one-to-one correspondence between
  each forward edge $(x, \beta, x\beta) \in \Eaf$ of $\ATree(y)$
  and each edge $(x, \beta, x\beta) \in \AEs$ of $\AST(y)$.
  Hence, what remains is to construct the backward edges in $\Eab$
  for $\ATree(y)$.
  A straightforward modification to our DAWG construction algorithm
  of Section~\ref{sec:DAWG_linear_const} can construct
  the backward edges of $\ATree(y)$;
  instead of working on the DAG $D$,
  we directly add the suffix links to the
  black nodes of $\AST(y)$ whose suffix links are not defined yet
  (namely, those that are neither branching nodes nor leaves
  of the suffix link tree $\SLT(y)$).
  Since the suffix links are reversed edges,
  by reversing them we obtain the backward edges of $\ATree(y)$.
  The labels of the backward edges can be easily computed in $O(n)$ time
  by storing in each node the length of the string it represents.
  Finally, we can sort the forward and backward edges in lexicographical order
  in overall $O(n)$ time, using the same idea as in
  Section~\ref{sec:DAWG_linear_const}.
\end{proof}

\section{Constructions of other text indexing structures}

In this section, we present how our algorithms for building
DAWGs and affix trees can be applied to constructing other text indexing data structures in linear time.

\subsection{DAWG for the reversed string}
\label{sec:DAWG_reversed}

Here we show that the algorithm that computes affix trees from Section~\ref{sec:atree_linear_const} can also be regarded as
a linear-time algorithm which computes the Weiner links for the forward and reversed strings simultaneously.
This will be a key to our efficient construction of other text indexing structures, which will be shown later in this section.

For the sake of simplicity, 
we here use a slightly different version of Weiner links than $\WL$ in Section~\ref{sec:folklore_algo}, which are called \emph{modified Weiner links} and are defined as follows:
$$\mWL = \{(x, a, ax) \mid  a \in \Sigma, x \in \Vs, ax \in \Substr(y)\}.$$
The difference from the normal Weiner links is that the destination nodes
of the modified Weiner links $ax$ can be \emph{implicit} (i.e. non-branching) nodes in the suffix tree (see the right diagram of Figure~\ref{fig:ATandWL}).
Hence, it is clear that there is a one-to-one correspondence between
$\WL$ and $\mWL$, and it is easy to convert $\mWL$ to $\WL$ and vice versa.

As in the case with (normal) Weiner links $\WL$, we call each modified Weiner link $(x, a, ax) \in \mWL$ \emph{explicit} if $ax \in \Vs$,
and implicit otherwise.
For each implicit modified Weiner link $(x, a, ax)$,
we consider an auxiliary nodes $ax$. 
We note that such auxiliary nodes $ax$ are
exactly the type-2 nodes of linear-size suffix tries~\cite{CrochemoreEGM16}.

The following theorem holds for the modified Weiner links.
\begin{theorem}
\label{theo:affix_WL}
For any edge $((x, \rev{x}), \rev{\alpha}, (\alpha x, \rev{x}\rev{\alpha})) \in \Eab$ of $\ATree(y)$,
$x \in \mathcal{L} \land \alpha x \notin \mathcal{L}$ iff
$(x, a, ax)$ is an implicit modified Weiner link of $\ST(y)$ where $a = \alpha[|\alpha|]$ (see Figure~\ref{fig:ATandWL}).
\end{theorem}

\begin{figure}[thb]
  \begin{center}
    \includegraphics[width=0.99\textwidth]{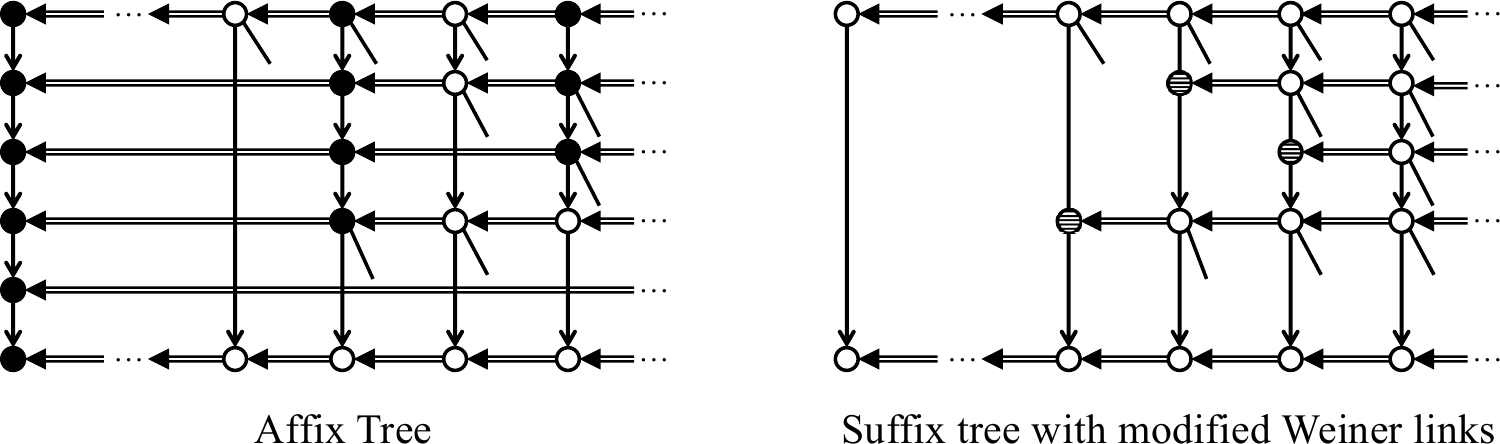}
  \end{center}
  \caption{Illustration for a relation between affix trees and modified Weiner links.
  The black circles represent black nodes $x \in \mathcal{L}$ of affix trees,
  while the striped circles represent auxiliary nodes $ax \notin \Vs$ for modified implicit Weiner link.}
  \label{fig:ATandWL}
\end{figure}

\begin{proof}
($\Leftarrow$)
Since $(x, a, ax)$ is an implicit modified Weiner link of $\ST(y)$,
$x$ is in $\mathcal{L}$ and $wax$ is not in $\mathcal{L}$ for any string $w \in \Sigma^*$.
Thus $x \in \mathcal{L} \land \alpha x \notin \mathcal{L}$.

\noindent ($\Rightarrow$)
By the definition of $\Eab$,
for any $i$~($1\leq i < |\alpha|$), we have $\rev{x}\cdot\rev{\alpha}[1..i]\notin \Va[2]$.
Thus $\alpha[|\alpha|]x \notin \mathcal{L}$.
Obliviously $ax \in \Substr(y)$, and thus
$(x, a, ax)$ is an implicit modified Weiner link of $\ST(y)$.
\end{proof}

By Theorem~\ref{theo:affix_WL}, when constructing the affix tree $\ATree(y)$ of the input string $y$,
the set of modified Weiner links $\mWL$ of $\ST(y)$ is also computed.
Because affix trees are symmetric data structures for the string and its reversal,
the modified Weiner links of the suffix tree $\ST(\rev{y})$ for the reversed string $\rev{y}$ are also computed at the same time.

\subsection{Symmetric CDAWGs}
\label{sec:symmetric_cdawg}

The \emph{Compact DAWG}~\cite{Blumer87} (\emph{CDAWG}) of string $y$,
denoted $\CDAWG(y)$ is an edge-labeled DAG that can be obtained by merging isomorphic subtrees of $\ST(y)$, or equivalently
by performing path-compressions to $\DAWG(y)$.
Each internal node of $\CDAWG(y)$ corresponds to a \emph{maximal repeat}
of the string $y$,
where a substring $x$ of $y$ is called a maximal repeat of $y$
if $x$ occurs in $y$ more than once,
and prepending or appending a character to $x$ decreases the number of its occurrences in $y$.
By the nature of maximal repeats,
$\CDAWG(y)$ for string $y$ and $\CDAWG(\hat{y})$ can share the same set of nodes.
The resulting data structure is yet another bidirectional text indexing structure, called \emph{symmetric CDAWG} for $y$.
Blumer et al.~\cite{Blumer87} showed the following:
\begin{theorem}[\cite{Blumer87}] \label{theo:symmetric_CDAWG}
  Given $\ST(y)$ and $\DAWG(\rev{y})$ for string $y$ which share the same nodes,
  one can build the symmetric CDAWG for string $y$ in linear time.
\end{theorem}

The algorithm by Blumer et al. in the above theorem does not use character comparisons.
Thus, by applying this method after building $\DAWG(\rev{y})$
over $\ST(y)$ using our technique from Section~\ref{sec:DAWG_reversed},
we obtain the following corollary:

\begin{corollary}
Given a string $y$ of length $n$ over an integer alphabet of polynomial size in $n$, we can compute edge-sorted symmetric CDAWG for $y$ in $O(n)$ time and space.
\end{corollary}

\subsection{Linear-size suffix tries}
\label{sec:linear-size_suffix_tries}

When reversed and unlabeled, the modified Weiner links are the suffix links of the so-called \emph{type-2 nodes} of the \emph{linear-size suffix tries}~\cite{CrochemoreEGM16}.
Let $\LSTrie(y)$ denote the linear-size suffix trie of string $y$.
Intuitively, $\LSTrie(y)$ is an edge-labeled rooted tree obtained from $\ST(y)$ as follows:
\begin{enumerate}
\item[(1)] For each \emph{reversed} and \emph{unlabeled} modified Weiner link $(ax, x)$, insert a non-branching explicit node $ax$ if $ax$ is not a branching node of $\ST(y)$. Such non-branching nodes are called type-2 nodes.
\item[(2)] Retain the first character in each edge label, and remove all the following characters along the edge.
\end{enumerate}
It is known that one can restore the whole edge label efficiently on $\LSTrie(y)$, without explicitly storing the input string $y$~\cite{CrochemoreEGM16,HendrianTI19,abs-2301-04295}.
Now the following corollary is immediate from Theorem~\ref{theo:affix_tree_linear_const} and Theorem~\ref{theo:affix_WL}:

\begin{corollary} \label{coro:linear_size_suffix_trie_linear_const}
Given a string $y$ of length $n$ over an integer alphabet of polynomial size in $n$, we can compute edge-sorted $\LSTrie(y)$ in $O(n)$ time and space.
\end{corollary}

\section{Computing Minimal Absent Words in $O(n + |\MAW(y)|)$ Time}
\label{sec:MAW}

As an application to our $O(n)$-time DAWG construction algorithm
of Section~\ref{sec:DAWG_linear_const},
in this section we show an optimal time algorithm
to compute the set of all minimal absent words of a given string
over an integer alphabet.

Finding minimal absent words of length 1 for a given string $y$
(i.e., the characters not occurring in $y$) is easy to do
in $O(n + \sigma)$ time and $O(1)$ working space
for an integer alphabet,
where $\sigma$ is the alphabet size.
In what follows, we concentrate on finding minimal absent words
of $y$ of length at least 2.

Crochemore et al.~\cite{crochemore98:_autom_forbid_words}
proposed a $\Theta(\sigma n)$-time algorithm to compute\\
$\MAW(y)$ for a given string $y$ of length $n$.
The following two lemmas, which show tight connections between $\DAWG(y)$ and $\MAW(y)$,
are implicitly used in their algorithm but under a somewhat different formulation.
Since our $O(n + |\MAW(y)|)$-time solution is built on the lemmas,
we give a proof for completeness.


\begin{figure}
 \centerline{\includegraphics[width=0.35\textwidth]{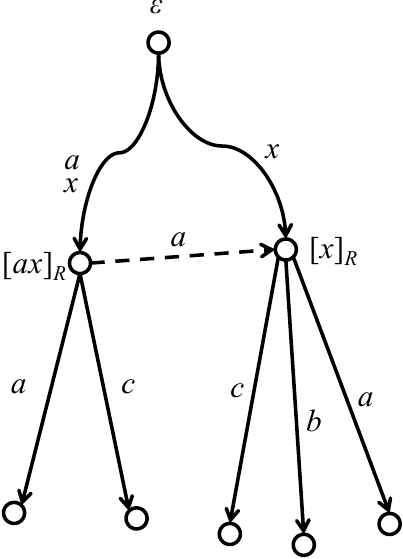}}
 \caption{Computing minimal absent words for the input string from the DAWG. In this case, $axb$ is a MAW since it does not occur in the string while $ax$ and $xb$ do.}
 \label{fig:daw}
\end{figure}

\begin{lemma}
\label{lem:daw}
Let $a,b \in \Sigma$ and $x\in \Sigma^*$.
If $axb\in \MAW(y)$, then $x = \Rrep{x}$,
namely, $x$ is the longest string represented by
node $\Reqc{x} \in \Vd$ of $\DAWG(y)$.
\end{lemma}
\begin{proof}
Assume on the contrary that $x \neq \Rrep{x}$.
Since $x$ is not the longest string of $\Reqc{x}$,
there exists a character $c \in \Sigma$
such that $cx \in \Substr(y)$ and $\Reqc{x} = \Reqc{cx}$.
Since $axb\in \MAW(y)$, it follows from Lemma~\ref{lem:maw}
that $xb \in \Substr(y)$.
Since $\Reqc{x} = \Reqc{cx}$,
$c$ always immediately precedes $x$ in $y$.
Thus we have $cxb \in \Substr(y)$.

Since $axb \in \MAW(y)$, $c \neq a$.
On the other hand, it follows from Lemma~\ref{lem:maw}
that $ax \in \Substr(y)$.
However, this contradicts that
$c$ always immediately precedes $x$ in $y$ and $c \neq a$.
Consequently, if $axb \in \MAW(y)$, then $x = \Rrep{x}$.
\end{proof}

For any node $v \in \Vd$ of $\DAWG(y)$ and character $b \in \Sigma$,
we write $\DAWGedge(v, b) = u$ if $(v, b, u) \in \Ed$ for some $u \in \Vd$,
and write $\DAWGedge(v, b) = \nil$ otherwise.
For any suffix link $(u, a, v) \in \Ld$ of $\DAWG(y)$,
we write $\DAWGlink(u) = v$. Note that there is exactly one
suffix link coming out from each node
$u \in \Vd$ of $\DAWG(y)$, so the character $a$ is unique for each node $u$.

\begin{lemma}
\label{lem:seed}
Let $a,b \in \Sigma$ and $x\in \Sigma^*$.
Then, $axb \in \MAW(y)$ iff
$x = \Rrep{x}$,
$\DAWGedge(\Reqc{x}, b) = \Reqc{xb}$,
$\DAWGlink(\Reqc{ax}) = \Reqc{x}$, and
$\DAWGedge(\Reqc{ax},b) = \nil$.
\end{lemma}

\begin{proof}
  ($\Rightarrow$)
  From Lemma~\ref{lem:daw}, $x = \Rrep{x}$.
  From Lemma~\ref{lem:maw}, $axb \not\in \Substr(y)$.
  However, $ax,xb \in \Substr(y)$,
  and thus we have
  $\DAWGedge(\Reqc{ax}, b) = \nil$,
  $\DAWGedge(\Reqc{x}, b) = \Reqc{xb}$, and
  $\DAWGlink(\Reqc{ax}) = \Reqc{x}$,
  where the last suffix link exists since $x = \Rrep{x}$.

  ($\Leftarrow$)
  Since $\DAWGedge(\Reqc{x},b) = \Reqc{xb}$
  and $\DAWGlink(\Reqc{ax}) = \Reqc{x}$,
  we have that $xb,ax \in \Substr(y)$.
  Since $ax \in \Substr(y)$ and $\DAWGedge(\Reqc{ax}, b) = \nil$,
  we have that $axb \not\in\Substr(y)$
  Thus from Lemma~\ref{lem:maw}, $axb \in \MAW(y)$.
\end{proof}

From Lemma~\ref{lem:seed} all MAWs of $y$ can be computed
by traversing all the states of $\DAWG(y)$ and comparing all out-going edges between nodes connected by suffix links.
See also Figure~\ref{fig:daw} for illustration.

A pseudo-code of the algorithm MF-TRIE by Crochemore et al.~\cite{crochemore98:_autom_forbid_words},
which is based on this idea, is shown in Algorithm~\ref{algo:cro}.
Since all characters in the alphabet $\Sigma$
are tested at each node, the total time complexity
becomes $\Theta(n\sigma)$. The working space is $O(n)$, since only
the DAWG and its suffix links are needed.

\begin{algorithm2e}[tb]
\caption{$\Theta(n\sigma)$-time algorithm (MF-TRIE) by Crochemore et al.~\cite{crochemore98:_autom_forbid_words}}
\label{algo:cro}
\KwIn{String $y$ of length $n$}
\KwOut{All minimal absent words for $y$}
\SetKw{AND}{and}
$\MAW \leftarrow \emptyset$\;
Construct $\DAWG(y)$ augmented with suffix links $\Ld$\;
\For{\Each non-source node $u$ of $\DAWG(y)$}{
  \For{\Each character $b \in \Sigma$\label{line:for_loop}}{
    \If{$\DAWGedge(u, b) = \nil$ \AND
      $\DAWGedge(\DAWGlink(u), b) \neq \nil$\label{line:condition}}{
        $\MAW \leftarrow \MAW \! \cup \! \{axb\}$\tcp*{$(u, a, \DAWGlink(u)) \! \in \! \Ld, x \! = \! \Longest(\DAWGlink(u))$}
    }
  }
}
Output $\MAW$\;
\end{algorithm2e}

Next we show that with a simple modification in the for loops of
the algorithm and with a careful
examination of the total cost, the set $\MAW(y)$ of all MAWs of
the input string $y$ can be computed
in $O(n+|\MAW(y)|)$ time and $O(n)$ working space.
Basically, the only change is to move the
``$\DAWGedge(\DAWGlink(u), b) \neq \nil$'' condition in Line~\ref{line:condition}
to the for loop of Line~\ref{line:for_loop}.
Namely, when we focus on node $u$ of $\DAWG(y)$,
we test only the characters which label
the out-edges from node $\DAWGlink(u)$.
A pseudo-code of the modified version is shown in Algorithm~\ref{algo:tsuji}.

\begin{algorithm2e}[tb]
\caption{Proposed $O(n+ |\MAW(y)|)$-time algorithm}
\label{algo:tsuji}
\KwIn{String $y$ of length $n$}
\KwOut{All minimal absent words for $y$}
$\MAW \leftarrow \emptyset$\;
Construct edge-sorted $\DAWG(y)$ augmented with suffix links $\Ld$\;
\For{\Each non-source node $u$ of $\DAWG(y)$}{
  \For{\Each character $b$ such that $\DAWGedge(\DAWGlink(u), b) \neq \nil$\label{line:for_loop_modified}}{
   \If{$\DAWGedge(u, b) = \nil$}{
    $\MAW \leftarrow \MAW \! \cup \! \{axb\}$\tcp*{$(u, a, \DAWGlink(u)) \! \in \! \Ld, x \! = \! \Longest(\DAWGlink(u))$}
    }
  }
}
Output $\MAW$\;
\end{algorithm2e}

\begin{theorem}
Given a string $y$ of length $n$ over an integer alphabet,
we compute $\MAW(y)$ in optimal $O(n + |\MAW(y)|)$ time
with $O(n)$ working space.
\end{theorem}

\begin{proof}
First, we show the correctness of our algorithm.
For any node $u$ of $\DAWG(y)$,
$\EndPos(\DAWGlink(u)) \supset \EndPos(u)$ holds
since every string in $\DAWGlink(u)$ is a suffix of the strings in $u$.
Thus, if there is an out-edge of $u$ labeled $c$,
then there is an out-edge of $\DAWGlink(u)$ labeled $c$.
Hence, the task is to find every character
$b$ such that there is an out-edge of $\DAWGlink(u)$ labeled $b$
but there is no out-edge of $u$ labeled $b$.
The for loop of Line~\ref{line:for_loop_modified}
of Algorithm~\ref{algo:tsuji} tests all such characters and only those.
Hence, Algorithm~\ref{algo:tsuji} computes $\MAW(y)$ correctly.

Second, we analyze the efficiency of our algorithm.
As was mentioned above,
minimal absent words of length 1 for $y$
can be found in $O(n + \sigma)$ time and $O(1)$ working space.
By Lemma~\ref{lem:MAW_bounds},
$\sigma \leq |\MAW(y)|$ and hence the $\sigma$-term is
dominated by the output size $|\MAW(y)|$.
Now we consider the cost of finding minimal absent words of length
at least 2 by Algorithm~\ref{algo:tsuji}.
Let $b$ be any character such that
there is an out-edge $e$ of $\DAWGlink(u)$ labeled $b$.
There are two cases:
(1) If there is no out-edge of $u$ labeled $b$,
then we output an MAW, so we can charge the cost to check $e$
to an output.
(2) If there is an out-edge $e'$ of $u$ labeled $b$,
then the trick is that we can charge the cost to check $e$ to $e'$.
Since each node $u$ has exactly one suffix link going out from it,
each out-edge of $u$ is charged only once in Case (2).
Since the out-edges of every node $u$ and those of $\DAWGlink(u)$
are both sorted, we can compute their difference
for every node $u$ in $\DAWG(y)$, in overall $O(n)$ time.
Edge-sorted $\DAWG(y)$ with suffix links can be
constructed in $O(n)$ time for an integer alphabet
as in Section~\ref{sec:DAWG_linear_const}.
Overall, Algorithm~\ref{algo:tsuji} runs in $O(n + |\MAW(y)|)$ time.
The space requirement is clearly $O(n)$.
\end{proof}

\section{Conclusions}\label{sec:conclusion}

In this paper, we proposed $O(n)$-time algorithms for constructing $\DAWG(y)$
and $\ATree(y)$ of a given string $y$ of length $n$,
over an integer alphabet of size polynomial in $n$.
These algorithms transform the suffix tree of $y$ to $\DAWG(y)$ and $\ATree(y)$.
We also showed how to build the symmetric CDAWG and linear-size suffix trie
for the input string over an integer alphabet.

As a further application of $\DAWG(y)$,
we presented an optimal $O(n + |\MAW(y)|)$-time algorithm to compute the set
$\MAW(y)$ of all minimal absent words of $y$ when the edge-sorted DAWG for $y$ is already computed.

\section*{Acknowledgements}

This work was supported by JSPS KAKENHI Grant Numbers
JP18J10967 (YF),
JP17H01697, JP26280003, JP22H03551 (SI), 
JP16H02783, JP20H04141 (HB),
JP25240003, JP18H04098 (MT).

\bibliographystyle{elsarticle-num}
\bibliography{ref}

\end{document}